\documentclass[final,onecolumn,12pt]{IEEEtran}
\usepackage{graphicx,psfrag,epsfig,epsf,latexsym,hhline,amsmath,amssymb,multirow}
\usepackage[usenames,dvipsnames]{pstricks}
\usepackage{pst-plot}
\usepackage{pstricks-add}
\usepackage{hyperref}
\usepackage{graphicx}
\usepackage{amsthm}
\usepackage{algorithm}
\usepackage{algpseudocode}
\usepackage[noadjust]{cite}
\usepackage{blindtext}
\usepackage{etoolbox}

\newcommand{\mc}[1]{\mathcal{#1}}

\newcommand{\mbb}[1]{\mathbb{#1}}

\newcommand{\msf}[1]{\mathsf{#1}}

\newcommand{\defeq}{\triangleq}
\newcommand{\Pp}{\mathbb{P}}

\newcommand{\Z}{\mathbb{Z}}

\newcommand{\Zw}{\mathbb{Z}[\omega]}
\newcommand{\Zi}{\mathbb{Z}[i]}

\newcommand{\iid}{i.\@i.\@d.\ }

\newcommand{\VOL}{\text{Vol}}

\theoremstyle{definition}\newtheorem{lemma}{Lemma}
\theoremstyle{definition}\newtheorem{proposition}[lemma]{Proposition}
\theoremstyle{definition}\newtheorem{theorem}[lemma]{Theorem}
\theoremstyle{definition}\newtheorem{corollary}[lemma]{Corollary}
\newtheorem{define}[lemma]{Definition}
\newtheorem{example}[lemma]{Example}

\newtheorem{remark}[lemma]{Remark}

\begin{document}
\title{Construction $\pi_A$ and $\pi_D$ Lattices: Construction, Goodness, and Decoding Algorithms}
\author{Yu-Chih Huang, \emph{Member, IEEE}, and Krishna R. Narayanan, \emph{Fellow, IEEE}
\thanks{The work of Y.-C. Huang was supported by the Ministry of Science and Technology, Taiwan, under Grant MOST 104-2218-E-305-001-MY2. The work of K. R. Narayanan was supported by the National Science Foundation under grant CCF-1302616. This paper was presented in part at the 2014 International Symposium on Information Theory \cite{huangisit14}, the 2014 Information Theory Workshop \cite{huang14ITW}, and the 2015 International Conference on Telecommunications \cite{huang15ict}.}

\thanks{Y.-C. Huang is with the Department of Communication Engineering, National Taipei University, 237 Sanxia District, New Taipei City, Taiwan (email: ychuang@mail.ntpu.edu.tw).}

\thanks{K. R. Narayanan is with the Department of Electrical and Computer Engineering, Texas A\&M University, College Station, TX 77843, USA (email: krn@tamu.edu).}
\thanks{Copyright (c) 2017 IEEE. Personal use of this material is permitted.  However, permission to use this material for any other purposes must be obtained from the IEEE by sending a request to pubs-permissions@ieee.org.}
}

\maketitle

\begin{abstract}
A novel construction of lattices is proposed. This construction can be thought of as a special class of Construction A {\black from codes over finite rings} that can be represented as the Cartesian product of $L$ linear codes over $\mathbb{F}_{p_1},\ldots,\mathbb{F}_{p_L}$, {\black respectively, and hence is referred to} as Construction $\pi_A$. The existence of a sequence of such lattices that is good for channel coding (i.e., Poltyrev-limit achieving) under multistage decoding is shown. {\black A new family of multilevel nested lattice codes based on Construction $\pi_A$ lattices is proposed and its achievable rate for the additive white Gaussian channel is analyzed}. A generalization named Construction $\pi_D$ is also investigated which subsumes Construction A with codes over prime fields, Construction D, and Construction $\pi_A$ as special cases.
\end{abstract}


\section{Introduction}
Lattices and codes based on lattices have been considered as one of the potential transmission schemes for point-to-point communications for decades.
Consider the additive white Gaussian (AWGN) channel
\begin{equation}
    \mathbf{y}=\mathbf{x}+\mathbf{z},
\end{equation}
where $\mathbf{y}$ is the received signal, $\mathbf{x}$ is the transmitted signal with input power constraint $P$, and $\mathbf{z}\sim\mc{N}(0,\eta^2\mathbf{I})$. In \cite{deBuda}, de Buda showed that one can use lattices shaped by a proper thick shell to reliably communicate under rates arbitrarily close to the channel capacity $\frac{1}{2}\log\left(1+\text{SNR}\right)$ bits/channel where $\text{SNR}\defeq P/\eta^2$ represents the signal-to-noise ratio. This result was then corrected by Linder \textit{et al.} \cite{linder93} which states that with de Buda's approach, only those lattice points {\black that} lie inside a thin spherical region are allowed to be used in order to achieve the rates promised by de Buda, which destroys the desired lattice structure. Urbanke and Rimoldi \cite{urbanke98} then showed that lattice codes with the minimum angle decoder can achieve the channel capacity. On the other hand, for lattice codes with lattice decoding, it was long believed that they can only achieve $\frac{1}{2}\log\left(\text{SNR}\right)$ bits/channel \cite{loeliger97}. In \cite{erez04}, Erez and Zamir finally showed that lattice codes can achieve the channel capacity with lattice decoding with the help of nested lattice shaping and an MMSE estimator at the receiver. Erez and Zamir's coding scheme is based on sequences of nested lattices that are constructed by Construction A \cite{LeechSloane71} \cite{conway1999sphere}\footnote{\black Here, the term ``Construction A" is used to represent Construction A with codes over prime fields. Later on, we will define Construction A with codes over rings. However, throughout the paper, for the sake of conciseness, we use this term to represent Construction A with codes over prime fields unless otherwise specified.}.

Recently, lattices have been adopted to many problems in network communications and their benefits have gone beyond merely practical aspects \cite{zamir_book}. In many networks (e.g. \cite{bresler10, nazer07GMAC, philosof11, wilson10, nazer2011CF}), it has been shown that the lattice structure enables one to exploit the structural gains induced by the channels and hence achieve higher rates than that provided by random codes. In most of these examples, the coding schemes are based on the random ensemble of nested lattice codes from Construction A by Erez and Zamir \cite{erez04}. {\black On one hand, this ensemble of nested lattice codes is known for its ability of producing capacity-achieving lattice codes and its structure which is suitable for many problems in network communications. On the other hand, decoding of a Construction A lattice typically depends on decoding the underlying linear code implemented over a prime field whose size has to be large in order to have a good lattice code \cite{erez04} \cite{erez05}.} This results in a large decoding complexity for lattices and codes based on them.

To alleviate this drawback, in this paper, we propose a novel lattice construction called Construction $\pi_A$ (previously called product construction in \cite{huang13ITW} a precursor of this paper) that can be thought of as a generalization of Construction A to codes which can be represented as the Cartesian product of $L$ linear codes over different prime fields $\mbb{F}_{p_1}$, $\ldots$, $\mbb{F}_{p_L}$. This generalization is enabled by a ring isomorphism between the product of prime fields and the quotient ring $\left(\mbb{Z}/\Pi_{l=1}^L p_l\mbb{Z}\right)$ which is guaranteed by the Chinese remainder theorem (CRT). Due to the multilevel nature, the Construction $\pi_A$ lattices admit multistage decoding which decodes the coset representatives level by level. This construction is then shown to be able to produce lattices that are good for channel coding (Poltyrev-limit achieving) under multistage decoding. This allows one to achieve the Poltyrev-limit with a substantially lower decoding complexity as now the complexity is dominated by the code over the prime field with the largest size rather than the product of them.



{\black Construction $\pi_A$ lattices are then adopted for communication over the AWGN channel. Following \cite{loeliger97}, we show the existence of lattice codes with sphere shaping that can achieve $\frac{1}{2}\log(\text{SNR})$ bits/channel with multistage decoding. We also tailor a recent construction of nested lattice codes by Ordentlich and Erez \cite{ordentlich_erez_simple} specifically for our Construction $\pi_A$ lattices. For such lattice codes, an isomorphism between lattice codewords and messages is guaranteed and can be easily identified. The achievable rate of the proposed multilevel nested lattice codes under multistage decoding is then analyzed. It is shown that with hypercube shaping, the proposed nested lattice codes only suffer from 1.53 dB SNR loss in shaping gain. This gap can be further reduced if there exists Construction $\pi_A$ lattices that can provide better shaping than that of hypercube shaping.}

Lattices generated by the Construction $\pi_A$ preserve most of the structure of Construction A lattices with codes over prime fields and hence can be applied to most of the applications using lattices from Construction A. However, there are some subtle differences between these constructions that may have a bearing on the application at hand. For e.g., lattices built from Construction $\pi_A$ appear to be ideally suited for lattice index coding \cite{natarajan15_LIC_ITW} \cite{natarajan14_LIC} and more so than other known constructions. Lattices from Construction $\pi_A$ have been considered for the compute-and-forward paradigm \cite{nazer2011CF} in \cite{huangisit14}. In this case, the set of integer combinations that can be decoded and forwarded may be smaller than those from Construction A since in effect the modulo operation at the relay is only over a ring instead of over a prime field.

We also provide a generalization of Construction $\pi_A$ to codes over rings. This generalization is called Construction $\pi_D$ and subsumes Construction A, Construction D \cite{BarnesSloane83} \cite[Page 232]{conway1999sphere}, and Construction $\pi_A$ as special cases. The main idea which allows this generalization is from the observation made in \cite{feng11rings} indicating the connection between Construction D and Construction A with codes over rings.

\subsection{Organization}
The paper is organized as follows. In Section~\ref{sec:prelim}, some background on lattices and algebra are provided together with a review and discussion about Construction A lattices. In Section~\ref{sec:prod_const}, we present the Construction $\pi_A$ lattices and show that such construction can produce good lattices. A detailed comparison between these lattices and Construction D lattices is provided in Section~\ref{sec:compare_D}. We then propose in Section~\ref{sec:const_pi_D} a generalization of the Construction $\pi_A$ lattices, which we refer to as the Construction $\pi_D$ lattices. Discussions about the Construction $\pi_A$ lattices are provided in Section~\ref{sec:discussion} followed by the proposed efficient decoding algorithms in Section~\ref{sec:para_dec}. In Section~\ref{sec:AWGN}, we consider using the Construction $\pi_A$ lattices for point-to-point communication over AWGN channel and propose a novel ensemble of nested multilevel lattice codes that can achieve the capacity under multistage decoding. Section~\ref{sec:conclusion} concludes the paper.

\subsection{Notations}
Throughout the paper, we use $\mbb{N}$, $\mbb{R}$, and $\mbb{C}$ to represent the set of natural numbers, real numbers, and complex numbers, respectively. $\mbb{Z}$, $\Zi$, and $\Zw$ are the rings of integers, Gaussian integers, and Eisenstein integers, respectively. We use $i\defeq \sqrt{-1}$ to denote the imaginary unit and define $\omega\defeq -\frac{1}{2}+i\frac{\sqrt{3}}{2}$. We use $\Pp(E)$ to denote the probability of the event $E$. Vectors and matrices are written in lowercase boldface and uppercase boldface, respectively. Random variables are written in Sans Serif font, {\black for example $\msf{X}$}. We use $\times$ to denote the Cartesian product and use $\oplus$ and $\odot$ to denote the addition and multiplication operations, respectively, over a finite ring/field where the ring/field size can be understood from the context if it is not specified.

\section{Preliminaries}\label{sec:prelim}
In this section, we briefly summarize background knowledge on lattices followed by some preliminaries on abstract algebra. For more details about lattices, lattice codes, and nested lattice codes, the reader is referred to \cite{erez04} \cite{erez05} \cite{conway1999sphere}. We then summarize the famous Construction A lattices.

\subsection{Lattices}
An $N$-dimensional lattice $\Lambda$ is a discrete subgroup of $\mathbb{R}^N$ which is closed under reflection and ordinary vector addition operation. i.e., $\forall \boldsymbol\lambda\in\Lambda$, we have $-\boldsymbol\lambda\in\Lambda$, and  $\forall \boldsymbol\lambda_1, \boldsymbol\lambda_2\in \Lambda$, we have $\boldsymbol\lambda_1 + \boldsymbol\lambda_2 \in \Lambda$. Some important operations and notions for lattices are defined as follows.
\begin{define}[Lattice Quantizer]
For a $\mathbf{x}\in\mathbb{R}^N$, the nearest neighbor quantizer associated with $\Lambda$ is denoted as
\begin{equation}
    Q_{\Lambda}(\mathbf{x})=\boldsymbol\lambda\in\Lambda;~\|\mathbf{x}-\boldsymbol\lambda\|\leq\|\mathbf{x}-\boldsymbol\lambda'\|~\forall\boldsymbol\lambda'\in\Lambda,
\end{equation}
where $\| .\|$ represents the $L_2$-norm operation and {\black the ties are broken arbitrarily}.
\end{define}

\begin{define}[Fundamental Voronoi Region]
The fundamental Voronoi region $\mathcal{V}_{\Lambda}$ is defined as
\begin{equation}
    \mathcal{V}_{\Lambda}=\{ \mathbf{x}: Q_{\Lambda}(\mathbf{x})=\mathbf{0} \}.
\end{equation}
\end{define}

\begin{define}[Modulo Operation]
The $\hspace{-3pt}\mod \Lambda$ operation returns the quantization error with respect to $\Lambda$ and is represented as
\begin{equation}
    \mathbf{x}\hspace{-3pt}\mod \Lambda = \mathbf{x}-Q_{\Lambda}(\mathbf{x}).
\end{equation}
\end{define}

The second moment of a lattice is defined as the average energy per dimension of a uniform probability distribution over $\mathcal{V}_{\Lambda}$ as
\begin{equation}
    \sigma^2(\Lambda) = \frac{1}{\text{Vol}(\mc{V}_{\Lambda})}\frac{1}{N}\int_{\mathcal{V}_{\Lambda}}\| \mathbf{x} \|^2 \mathrm{d}\mathbf{x},
\end{equation}
where $\text{Vol}(\mc{V}_{\Lambda})$ is the volume of $\mathcal{V}_{\Lambda}$. The normalized second moment of the lattice is then defined as
\begin{equation}
    G(\Lambda) = \frac{\sigma^2(\Lambda)}{\text{Vol}(\mc{V}_{\Lambda})^{2/N}},
\end{equation}
which is lower bounded by that of a sphere which asymptotically approaches $\frac{1}{2\pi e}$ in the limit as $N\rightarrow \infty$. Note that $G(\Lambda)$ is invariant to scaling.

We now define two important notions of goodness for lattices.
\begin{define}[Goodness for MSE Quantization]
    We say that a sequence of lattices is asymptotically good for MSE quantization if
    \begin{equation}
        \underset{N\rightarrow\infty}{\lim} G(\Lambda) = \frac{1}{2\pi e}.
    \end{equation}
\end{define}
Consider the unconstrained AWGN channel $\mathbf{y}=\mathbf{x}+\mathbf{z}$ where $\mathbf{x}$, $\mathbf{y}$, and $\mathbf{z}\sim \mathcal{N}(0,\eta^2\cdot I)$ represent the transmitted signal, the received signal, and the noise, respectively. Moreover, let $\mathbf{x} \in \Lambda$ and let there be no power constraint on $\mathbf{x}$ so that any lattice point could be sent.
\begin{define}[Goodness for Channel Coding]
    We say that a sequence of lattices is asymptotically good for channel coding if whenever
    \begin{equation}\label{eqn:poltyre_good}
        \eta^2<\frac{\text{Vol}(\mc{V}_{\Lambda})^{\frac{2}{N}}}{2\pi e},
    \end{equation}
    the error probability of decoding $\mathbf{x}$ from $\mathbf{y}$ can be made arbitrarily small as $N$ increases.
\end{define}
Here, by goodness for channel coding, we particularly mean a sequence of lattices that approach the Poltyrev limit defined in \eqref{eqn:poltyre_good}. There is a stronger version of Poltyrev-goodness stating that the sequence of lattices achieves an error exponent lower bounded by the Poltyrev exponent \cite{poltyrev94}. However, the proof of achieving Poltyrev exponent is more involved and we do not pursue it in this paper. The interested reader is referred to \cite{poltyrev94} and \cite{erez04}.

\subsection{Algebra}
In this subsection, we provide some preliminaries that will be useful in explaining our results in the following sections. All the lemmas are provided without proofs for the sake of brevity; however, their proofs can be found in standard textbooks on abstract algebra, see for example \cite{Hungerford74}.

We first recall some basic definitions for commutative rings. Let $\mc{R}$ be a commutative ring. Let $a, b\neq 0 \in\mc{R}$ but $ab = 0$, then $a$ and $b$ are \textit{zero divisors}. If $ab = ba = 1$, then we say $a$ is a \textit{unit}. Two elements $a, b\in\mc{R}$ are associates if $a$ can be written as the multiplication of a unit and $b$. A non-unit element $\phi\in\mc{R}$ is a prime if whenever $\phi$ divides $ab$ for some $a, b \in \mc{R}$, either $\phi$ divides $a$ or $\phi$ divides $b$. An \textit{integral domain} is a commutative ring with identity and no zero divisors. An additive subgroup $\mc{I}$ of $\mc{R}$ satisfying $ar\in\mc{I}$ for $a\in\mc{I}$ and $r\in\mc{R}$ is called an \textit{ideal} of $\mc{R}$. An ideal $\mc{I}$ of $\mc{R}$ is proper if $\mc{I}\neq\mc{R}$. An ideal generated by a singleton is called a \textit{principal ideal}. A \textit{principal ideal domain} (PID) is an integral domain in which every ideal is principal. Famous and important examples of PID include $\mbb{Z}$, $\Zi$ and $\Zw$. Let $a, b\in\mc{R}$ and $\mc{I}$ be an ideal of $\mc{R}$; then $a$ is congruent to $b$ \textit{modulo} $\mc{I}$ if $a-b\in\mc{I}$. The quotient ring $\mc{R}/\mc{I}$ of $\mc{R}$ by $\mc{I}$ is the ring with addition and multiplication defined as
\begin{align}
    (a+\mc{I})+(b+\mc{I}) &= (a+b)+\mc{I}, \text{~and} \\
    (a+\mc{I})\cdot(b+\mc{I}) &= (a\cdot b)+\mc{I}.
\end{align}

A proper ideal $\mc{P}$ of $\mc{R}$ is said to be a \textit{prime ideal} if for $a, b\in\mc{R}$ and $ab\in\mc{P}$, either $a\in\mc{P}$ or $b\in\mc{P}$. {\black For two ideals $\mc{I}_1$ and $\mc{I}_2$ of $\mc{R}$, let us define
\begin{equation}
    \mc{I}_1+\mc{I}_2 \triangleq \{a+b: a\in\mc{I}_1, b\in\mc{I}_2\},
\end{equation}
and
\begin{equation}
    \mc{I}_1\mc{I}_2 \triangleq \left\{\sum_{j=1}^n a_j b_j: a_j\in\mc{I}_1, b_j\in\mc{I}_2, n\in\mbb{N}\right\}.
\end{equation}
$\mc{I}_1$ and $\mc{I}_2$ are \textit{relatively prime} if $\mc{R} = \mc{I}_1+\mc{I}_2$, which also implies that $\mc{I}_1\mc{I}_2 = \mc{I}_1\cap\mc{I}_2$.} A proper ideal $\mc{O}$ of $\mc{R}$ is said to be a \textit{maximal ideal} if $\mc{O}$ is not contained in any strictly larger proper ideal. It should be noted that every maximal ideal is also a prime ideal but the reverse may not be true. Let $\mc{R}_1, \mc{R}_2, \ldots, \mc{R}_L$ be a family of rings, the direct product of these rings, denoted by $\mc{R}_1\times \mc{R}_2\times \ldots \times\mc{R}_L$, is the direct product of the additive Abelian groups $\mc{R}_l$ equipped with multiplication defined by the \textit{componentwise} multiplication.

Let $\mc{R}_1$ and $\mc{R}_2$ be rings. A function $\sigma:\mc{R}_1\rightarrow \mc{R}_2$ is a \textit{ring homomorphism} if
\begin{align}
    \sigma(1) &= 1, \\
    \sigma(a + b) &= \sigma(a) \oplus \sigma(b) ~\forall a,b\in\mc{R}_1, \\
    \sigma(a\cdot b) &= \sigma(a)\odot \sigma(b),~\forall a,b\in\mc{R}_1.
\end{align}
A homomorphism is said to be an \textit{isomorphism} if it is bijective. It is worth mentioning that for an ideal $\mc{I}$, $\hspace{-3pt}\mod\mc{I}:\mc{R} \rightarrow \mc{R}/\mc{I}$ is a natural ring homomorphism. A $\mc{R}$-module $\mc{N}$ over a ring $\mc{R}$ consists of an Abelian group ($\mc{N},+$) and an operation $\mc{R}\times \mc{N}\rightarrow \mc{N}$ which satisfies the same axioms as those for vector spaces. Let $\mc{N}_1$ and $\mc{N}_2$ be $\mc{R}$-modules. A function $\varphi:\mc{N}_1\rightarrow \mc{N}_2$ is a \textit{$\mc{R}$-module homomorphism} if
\begin{align}
    \varphi(a + b) &= \varphi(a) \oplus \varphi(b) ~\forall a,b\in\mc{N}_1 \text{~and} \\
    \varphi(r a) &= r \varphi(a),~\forall r\in\mc{R}, a\in\mc{N}_1.
\end{align}

We now present some lemmas which serve as the foundation of the paper.
\begin{lemma}\label{lma:PID}
    If $\mc{R}$ is a PID, then every non-zero prime ideal is maximal.
\end{lemma}

\begin{lemma}\label{lma:MAX}
    Let $\mc{I}$ be an ideal in a commutative ring $\mc{R}$ with identity $1_{\mc{R}}\neq 0$. If $\mc{I}$ is maximal and $\mc{R}$ is commutative, then the quotient ring $\mc{R}/\mc{I}$ is a field.
\end{lemma}

\begin{lemma}[Chinese Remainder Theorem]\label{lma:CRT}
    Let $\mc{R}$ be a commutative ring, and $\mc{I}_1,\ldots,\mc{I}_n$ be {\black relatively prime} ideals in $\mc{R}$. Then,
    \begin{equation}
        \mc{R}/\cap_{i=1}^n\mc{I}_i \cong \left(\mc{R}/\mc{I}_1\right)\times\ldots\times\left(\mc{R}/\mc{I}_n\right).
    \end{equation}
\end{lemma}

\begin{example}\label{exp:ring_iso}
    Consider the PID $\mbb{Z}$ and one of its ideal $6\mbb{Z}$. Note that one can do the prime factorization $6=2\cdot 3$. Now since 2 and 3 are primes, $2\mbb{Z}$ and $3\mbb{Z}$ are prime ideals. Also, since $2\mbb{Z}+3\mbb{Z}=\mbb{Z}$, they are relatively prime. This implies that $2\cdot 3\mbb{Z} = 2\mbb{Z}\cap 3\mbb{Z}$. One has that
    \begin{align}
        \mbb{Z}_6 &\cong \mbb{Z}/6\mbb{Z} = \mbb{Z}/2\cdot 3 \mbb{Z} \nonumber \\
        &\overset{(a)}{=} \mbb{Z}/2\mbb{Z}\cap 3\mbb{Z} \nonumber \\
        &\overset{(b)}{\cong} \mbb{Z}/2\mbb{Z}\times\mbb{Z}/3\mbb{Z} \nonumber \\
        &\overset{(c)}{\cong} \mbb{F}_2\times \mbb{F}_3,
    \end{align}
    where (a) follows from that $2\mbb{Z}$ and $3\mbb{Z}$ are relatively prime, (b) follows from CRT, and (c) is from Lemma~\ref{lma:MAX}. {\black An isomorphism is given by $\mc{M}(v^1,v^2)=p_2v^1+(-1)p_1v^2\hspace{-3pt}\mod p_1p_2\mbb{Z}=3v^1-2v^2 \hspace{-3pt}\mod 6\mbb{Z}$ where $v^1\in\mbb{F}_2$ and $v^2\in\mbb{F}_3$.}
    One can easily see from this example that the product of two fields is not a field. In this example, the product is isomorphic to $\mbb{Z}_6$ which is a ring but not a field.
\end{example}

\subsection{Construction A Lattices}
We now review Construction A lattices and discuss some properties of such lattices and some related constructions. For the sake of brevity, we only discuss Construction A lattices over $\mbb{Z}$ but extensions to other PID such as $\Zi$ and $\Zw$ are possible (see for example \cite{Engin14}). It is worth noting that although Construction A from codes over prime fields is more frequently seen in the literature, we consider here the more general definition {\black of Construction A of lattices from codes over a finite rings $\mbb{Z}_q$ ($q$-ary lattices in \cite{zamir_book})}. This more general construction subsumes {\black the construction proposed in Section~\ref{sec:prod_const}} as a special case.

{\black {\bf \underline{Construction A}} \cite{LeechSloane71} \cite{conway1999sphere} \cite[Page 31]{zamir_book} Let $q>1$ be an integer. Let $k,N\in\mbb{N}$ be integers such that $k\leq N$ and let $\mathbf{G}$ be a generator matrix of an $(N,k)$ linear code over $\mbb{Z}_q$. Construction A consists of the following steps:
\begin{enumerate}
    \item {\black Consider the linear code $C=\{\mathbf{x}=\mathbf{G}\odot\mathbf{y}:\mathbf{y}\in\mbb{Z}_q^k\}$}, where all operations are over $\mbb{Z}_q$.
    \item {\black ``Expand" $C$ to a lattice in $\mbb{Z}^N$ defined as:
    \begin{equation}
        \Lambda_{\text{A}} \defeq \left\{ \mathbf{x}\in\mbb{Z}^N: \mathbf{x}\hspace{-3pt}\mod q\in C \right\} = C + q\mbb{Z}.
    \end{equation}}
\end{enumerate}
{\black It is shown in \cite[page 31]{zamir_book} that $\Lambda_{\text{A}}$ is a non degenerated lattice, $q\mbb{Z}^N\subset\Lambda\subset\mbb{Z}^N$, and that the volume of this lattice is $q^N/M$, where $M$ is the size of the code $C$.}

{\black Using lattices from Construction A with codes over $\mbb{F}_p$ for communication over the AWGN channel has been investigated for decades.} It has been shown by Forney \textit{et al.} \cite{forney2000} that Construction A yields a sequence of lattices that is good for channel coding whenever the underlying linear codes achieve the capacity of the corresponding $\hspace{-3pt}\mod p\mbb{Z}$-channel and $p$ is sufficiently large (tends to infinity). Loeliger in \cite{loeliger97} used the Minkowski-Hlawka theorem to show that randomly picking {\black a code from the random $(N,k)$ linear code ensemble} and applying the above construction would with high probability result in lattices that are good for channel coding if $p$ tends to infinity. Using such lattices for the power-constrained AWGN channel, Loeliger showed that $\frac{1}{2}\log(\text{SNR})$ is achievable. Erez \textit{et al.} \cite{erez05} then moved on and showed that the random ensemble of Loeliger in fact produces lattices that are simultaneously good in many senses including channel coding, MSE quantization, covering, and packing with high probability if the parameters $p$, $k$, $N$ are carefully chosen (and of course tend to infinity). This result has allowed Erez and Zamir to show the existence of a sequence of nested lattice codes generated from Construction A lattices that can achieve the AWGN capacity, $\frac{1}{2}\log(1+\text{SNR})$ bits/channel, under lattice decoding \cite{erez04}. Since then, the nested lattice code ensemble of Erez and Zamir has been applied to many problems in networks. It is worth noting that recently, there has been another ensemble of nested lattice codes from Construction A lattices proposed by Ordentlich and Erez \cite{ordentlich_erez_simple} which can achieve the AWGN capacity as well.

From the practical aspect, there have been some efforts in constructing lattices based on Construction A with practical coding schemes. In \cite{diPietro12} (also appeared in \cite{Engin12}), di Pietro \textit{et al.} used non-binary low-density parity-check (LDPC) codes in conjunction with Construction A to construct lattices and referred this family of lattices to as the low-density A (LDA) lattices. Simulation results reported in \cite{diPietro12} showed that such lattices can approach the Poltyrev-limit to within 0.7 dB at a block length of 10000 under message-passing decoding. They then moved on and rigorously showed in \cite{diPietro13} that LDA lattices can achieve the Poltyrev-limit under maximum likelihood decoding. Inspired by the success of spatially-coupled LDPC codes for binary memoryless channels, Tunali \textit{et al.} \cite{Engin13ITW} replaced LDPC codes in LDA lattices by spatially-coupled LDPC codes and reported a BP-threshold of 0.19 dB away from the Poltyrev-limit at a block length of $1.29\times 10^6$. {\black Very recently, it has been shown in \cite{diPietro16} that LDA lattice codes can achieve the AWGN capacity without dithering. It is worth noting that there is another ensemble of lattice codes inspired by LDPC codes called low-density lattice codes (LDLC) \cite{LDLC}. It has been shown by simulation that LDLC can provide good error probability performance with low-complex decoding \cite{Kurkoski08, Kurkoski10}; however, to the best of our knowledge, no goodness results have been shown for lattice codes drawn from this ensemble. }

One crucial issue in Construction A lattices with codes over $\mbb{F}_p$ is that typically speaking, the decoding complexity depends on decoding the underlying linear code, which is over $\mbb{F}_p$. However, in order to get a good lattice, one has to let $p$ grow rapidly; hence this results in a huge decoding complexity. For instance, the 0.7 dB gap result reported in \cite{diPietro12} corresponds to using a linear code over $\mbb{F}_{41}$ together with a prime ideal $(4+5i)\Zi$ in $\Zi$ and the 0.19 dB result in \cite{Engin13ITW} corresponds to a linear code over $\mbb{F}_{31}$ with a prime ideal $(-1-6\omega)\Zw$ in $\Zw$. This is mainly because after identifying the coset representative (i.e., decoding the underlying linear code), lattice points inside a coset are unprotected by any code and therefore the only obvious way to avoid errors is to increase the Euclidean distance, i.e., to increase $p$.

\section{Construction $\pi_A$ Lattices}\label{sec:prod_const}
{\black In order to alleviate the high decoding complexity of Construction A lattices, we propose a lattice construction called Construction $\pi_A$. This construction is a special case of Construction A from codes over rings.} Note that Construction $\pi_A$ can be used for generating lattices over $\mbb{Z}$, $\Zi$, and $\Zw$. In this section, we will only talk about $\mbb{Z}$ and the cases of $\Zi$ and $\Zw$ will follow similarly. A depiction of Construction $\pi_A$ can be found in Fig.~\ref{fig:lattice_const}. Construction $\pi_A$ heavily relies on the existence of ring isomorphisms guaranteed by CRT in Lemma~\ref{lma:CRT} (see also \cite[Corollary 2.27]{Hungerford74}).

\begin{proposition}\label{prop:ring_iso}
    Let $p_1, p_2,\ldots,p_L$ be a collection of distinct primes and let $q=\Pi_{l=1}^L p_l$. There exists a ring isomorphism $\mc{M}:\times_{l=1}^L \mbb{F}_{p_l}\rightarrow \mbb{Z}/ q \mbb{Z}$. 
\end{proposition}
\begin{proof}
    Note that
    \begin{align}\label{eqn:isomorphic}
        \mbb{Z}_q\cong \mbb{Z}/q \mbb{Z} &\overset{(a)}{\cong} \mbb{Z}/\cap_{l=1}^L p_l\mbb{Z} \nonumber \\
        &\overset{(b)}{\cong} \mbb{Z}/p_1\mbb{Z}\times\ldots\times\mbb{Z}/p_L\mbb{Z} \nonumber \\
        &\overset{(c)}{\cong} \mbb{F}_{p_1}\times\ldots\times\mbb{F}_{p_L},
    \end{align}
    where (a) follows from that $p_l \mbb{Z}$ are relatively prime, (b) is from CRT in Lemma~\ref{lma:CRT}, and (c) is due to the fact that $\mbb{Z}$ is a PID and Lemma~\ref{lma:MAX}. Therefore, a ring isomorphism $\mc{M}$ between the product of fields $\times_{l=1}^L \mbb{F}_{p_l}$ and the quotient ring $\mbb{Z}/\Pi_{l=1}^L p_l\mbb{Z}$ exists. 
\end{proof}
{\black One way to obtain a ring isomorphism $\mc{M}$ is to first label every element $\zeta\in\mbb{Z}_q$, $q=\Pi_{l=1}^L p_l$, by the natural mapping and then define $\mc{M}^{-1}\defeq\left(\zeta\hspace{-3pt}\mod p_1,\ldots,\zeta\hspace{-3pt}\mod p_L\right)$. Another way is to directly solve for $a_1,\ldots,a_L$ in B\'{e}zout's identity given by
\begin{equation}
    a_1q_1+a_2q_2+\ldots+a_L q_L=1,
\end{equation}
where $q_l=q/p_l$ and obtain
\begin{equation}\label{eqn:ring_iso}
    \mc{M}(v^1,\ldots,v^L) = a_1q_1v^1+a_2q_2v^2+\ldots+a_Lq_Lv^L\hspace{-3pt}\mod q,
\end{equation}
where $v^l\in\mbb{F}_{p_l}\cong\mbb{Z}_{p_l}$.} 

We are now ready to present the Construction $\pi_A$ lattices.

\textbf{\underline{Construction $\pi_A$}} Let $p_1, p_2,\ldots, p_L$ be distinct primes. Let $m^l$, $N$ be integers such that $m^l\leq N$ and let $\mathbf{G}^l$ be a generator matrix of an $(N,m^l)$ linear code over $\mbb{F}_{p_l}$ for $l\in\{1,\ldots,L\}$. Construction $\pi_A$ consists of the following steps,

\begin{enumerate}
    \item Define the discrete codebooks $C^l=\{\mathbf{x}=\mathbf{G}^l\odot\mathbf{y}:\mathbf{y}\in\mbb{F}_{p_l}^{m^l}\}$ for $l\in\{1,\ldots,L\}$.
    \item Construct $\Lambda^* \defeq \mc{M}(C^1,\ldots,C^L)$ where $\mc{M}:\times_{l=1}^L \mbb{F}_{p_l}\rightarrow \mbb{Z}/\Pi_{l=1}^L p_l \mbb{Z}$ is a ring isomorphism.
    \item Tile $\Lambda^*$ to the entire $\mbb{R}^N$ to form $\Lambda \triangleq  \Lambda^* + \Pi_{l=1}^L p_l \mbb{Z}^N$.\footnote{Note that scaling by real numbers does not change the structure of a lattice; therefore, throughout the paper, we use $\Lambda \triangleq  \Lambda^* + \Pi_{l=1}^L p_l \mbb{Z}^N$ and $\Lambda \triangleq  \left(\Pi_{l=1}^L p_l\right)^{-1}\Lambda^* + \mbb{Z}^N$ interchangeably.}
\end{enumerate}
\begin{figure}
    \centering
    \includegraphics[width=2.4in]{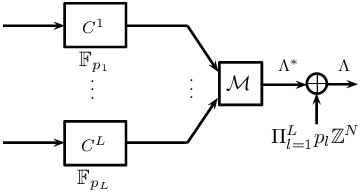}
    \caption{The Construction $\pi_A$ of lattices.}
    \label{fig:lattice_const}
\end{figure}
Note that the existence of the ring isomorphism in step 2) is guaranteed by Proposition~\ref{prop:ring_iso}. {\black Now, let $q=\Pi_{l=1}^L p_l$. Every Construction $\pi_A$ lattice is a Construction A lattice with a linear code over $\mbb{Z}_q$ with the generator matrix $\mathbf{G}$ such that $\mathbf{G}\hspace{-3pt}\mod p_l$ and $\mathbf{G}^l$ generate the same code $C^l$. Consequently, similar to \cite[Proposition 2.5.1]{zamir_book}, the following properties hold.}

\begin{proposition}\label{prop:lattice}

    \begin{enumerate}
        \item $\Lambda$ is a lattice and $\boldsymbol\lambda\in\Lambda$ if and only if $\sigma(\boldsymbol\lambda)\in C^1\times \ldots \times C^L$ where $\sigma = \mc{M}^{-1}\circ\hspace{-3pt}\mod\Pi_{l=1}^L p_l\mbb{Z}$ is the ring homomorphism given in Proposition~\ref{prop:ring_iso}.
        \item $q\mbb{Z}^N\subset\Lambda\subset\mbb{Z}^N$.
        \item $\text{Vol}(\mc{V}_{\Lambda})=\left|\mbb{Z}^N/\Lambda\right|=q^N/M$, where $M$ is the size of $C^1\times\ldots\times C^L$. Furthermore, if every $\mathbf{G}^l$ is full rank,  $\text{Vol}(\mc{V}_{\Lambda})= \Pi_{l=1}^L p_l^{N-m^l}$.
    \end{enumerate}
\end{proposition}

\begin{example}
Let us consider a two-level example where $p_1=3$ and $p_2=5$. One has $\Z/15\Z\cong \mbb{F}_3\times \mbb{F}_5$ and a ring isomorphism {\black $\mc{M}(v^1,v^2)=-p_2v^1+2p_1v^2\hspace{-3pt}\mod p_1p_2=-5v^1+6v^2\hspace{-3pt}\mod 15$, where $v^1\in\mbb{F}_3$ and $v^2\in\mbb{F}_5$.} Let us choose $\mathbf{G}^1=[1,2]^T$ over $\mbb{F}_3$ and $\mathbf{G}^2=[1,1]^T$ over $\mbb{F}_5$ which define the discrete codebooks $C^1$ and $C^2$, respectively. In Fig.~\ref{fig:prod_ex1}, we show the step 2) of Construction $\pi_A$ where we use the ring isomorphism $\mc{M}$ to modulate the codewords onto $\Z/15\Z$ to form $\Lambda^*$. Each element $\boldsymbol\lambda^*\in\Lambda^*$ is a coset representative of the coset $\boldsymbol\lambda^* + 15 \Z^2$. In Fig.~\ref{fig:prod_ex2}, we further tile this set of coset representatives to the entire $\mbb{R}^2$; this corresponds to the step 3) in Construction $\pi_A$. {\black It should be noted that, using the above ring isomorphism, the lattice considered in this example is identical to the Construction A lattice from the linear code over $\mbb{Z}_{15}$ with $\mathbf{G}=[1,11]$ (see Fig.~\ref{fig:prod_ex2}).}
\begin{figure}
    \centering
    \includegraphics[width=3.5in]{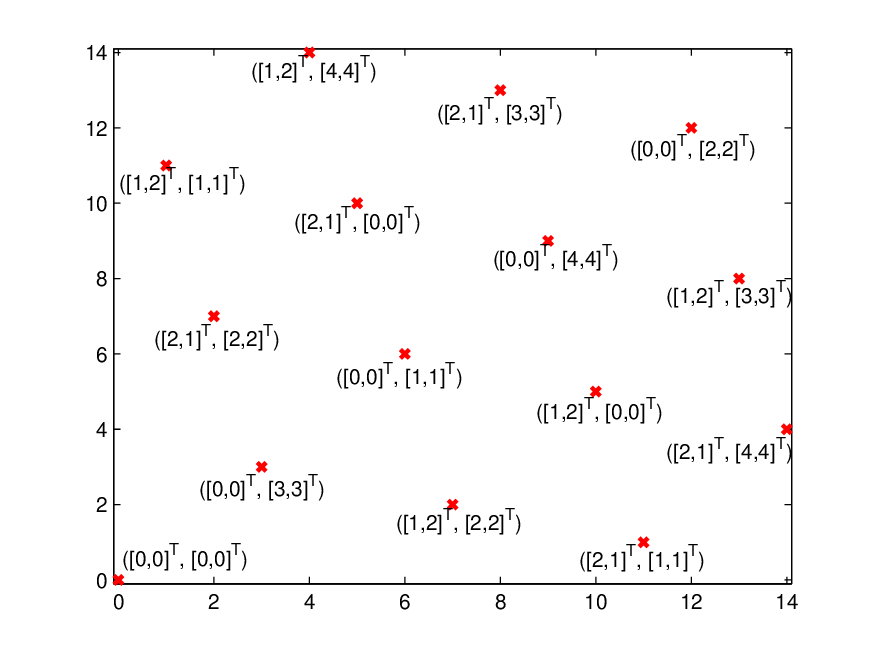}
    \caption{The set of coset representatives $\Lambda^*$ generated by $\mc{M}(C^1,C^2)$ where the corresponding codewords are also shown.}
    \label{fig:prod_ex1}
\end{figure}
\begin{figure}
    \centering
    \includegraphics[width=3.5in]{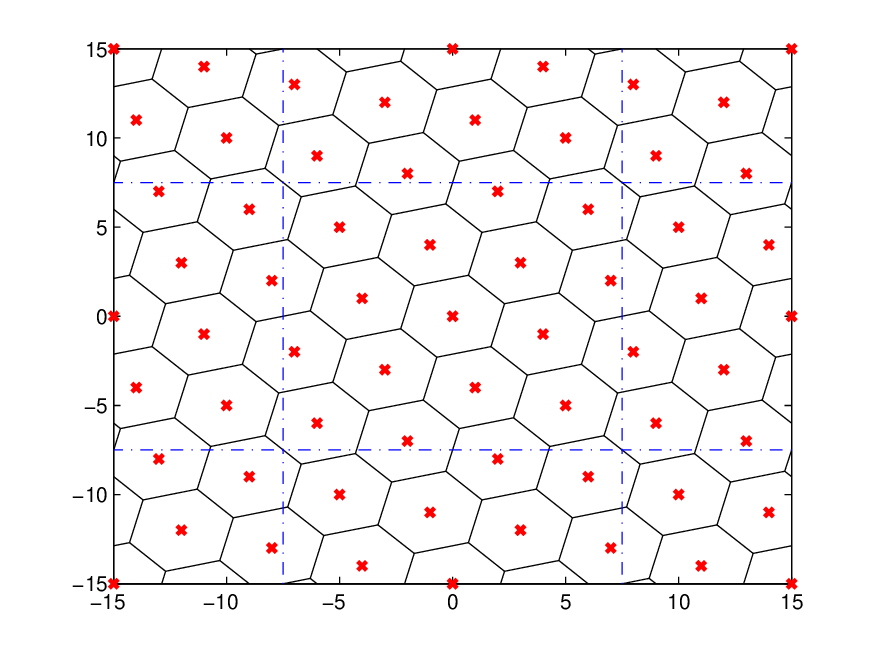}
    \caption{Tiling $\Lambda^*$ to the entire $\mbb{R}^2$ to form $\Lambda$. i.e., $\Lambda = \Lambda^*+15\Z^2$.}
    \label{fig:prod_ex2}
\end{figure}
\end{example}

One important reason that makes Construction $\pi_A$ lattices distinguish themselves from other Construction A lattices with codes over $\mbb{Z}_q$ is the close connection to multilevel coding over prime fields $\mbb{F}_{p_l}$. This structure is exploited to show the following properties.
{\black \begin{theorem}\label{thm:goodness}
    There exists a sequence of Construction $\pi_A$ lattices that is good for channel coding under multistage decoding.
\end{theorem}

\begin{IEEEproof}
    See Appendix~\ref{apx:lattice_goodness}. This proof closely follows the steps by Forney in \cite{forney2000} instead of the Loeliger's proof in \cite{loeliger97}. 
\end{IEEEproof}
}

\begin{remark}
    When proving the Poltyrev-goodness, unlike Construction A lattices with codes over $\mbb{F}_p$ letting $p\rightarrow\infty$ and Construction D lattices \cite{BarnesSloane83} \cite{forney2000} letting $L\rightarrow\infty$, for Construction $\pi_A$, we let $\Pi_{l=1}^L p_l\rightarrow \infty$. Thus, $L$ and $p_1,p_2,\ldots,p_L$ are parameters that can be chosen. This construction allows us to achieve the Poltyrev-limit with a significantly lower decoding complexity compared to Construction A lattices as now the complexity is not determined by the number of elements in $\Lambda^*$ but by the greatest prime divisor in the prime factorization of $|\Lambda^*|$. However, the complexity is higher than that of the Construction D lattices whose complexity is always determined by coding over $\mbb{F}_2$. This is a direct consequence of the fact that all primes should be distinct in Construction $\pi_A$.
\end{remark}

\begin{remark}
    The possible values of $|\Lambda^*|$ are confined in a subset of $\mbb{N}$. For example, for $\mbb{Z}$, Constellation $\pi_A$ allows $|\Lambda^*|$ to be any square-free integer \cite{sloaneOEIS}. Nonetheless, the choices of such $|\Lambda^*|$ are very rich and absorb Construction A lattices as special cases. There are many interesting things one can say about square-free integers. For example, the asymptotic density of square-free integers in $\mbb{Z}$ is $6/\pi^2\approx 0.6079$ which indicates that for a large portion of elements (infinitely many) in $\mbb{N}$, Construction $\pi_A$ can be used. Moreover, the asymptotic density of primes $p$ such that $p-1$ is square-free equals the Artin constant $A\approx 0.3739$ which indicates that for a large portion (infinitely many) of $p$, Construction A over $p$ can be replaced by Construction $\pi_A$ over $p-1$ to for reducing decoding complexity at a cost of a very small rate loss. The interested reader is referred to \cite{MoreeHommerson03}.
\end{remark}

{\black So far, we have only talked about lattices instead of lattice codes. In what follows, we use the proposed multilevel lattices in conjunction with sphere shaping for transmission over the AWGN channel. We follow the steps of Loeliger \cite{loeliger97} and obtain the following corollary.
\begin{corollary}
    For the AWGN channel, there exists a sequence of lattice codes that can achieve $R=\frac{1}{2}\log\left(\frac{P}{\eta^2}\right)$ under multistage decoding.
\end{corollary}
\begin{IEEEproof}
    Let $\mc{S}\subset\mbb{R}^n$ be the spherical shaping region with radius $\sqrt{NP}$ and let $\Lambda'$ be a good construction $\pi_A$ lattice from Theorem~\ref{thm:goodness}. Let $M$ be the desired number of codewords in the codebook. We scale $\Lambda'$ to obtain $\Lambda$ such that $\VOL(\mc{V}_{\Lambda})=\VOL(\mc{S})/M$. Note that scaling does not ruin the goodness for channel coding and thus $\Lambda$ is also good for channel coding.
    From \cite[Lemma 2]{loeliger97}, there exists a $\mathbf{v}\in\mbb{R}^N$ such that $|(\mathbf{v}+\Lambda)\cap\mc{S}|\geq \VOL(\mc{S})/\VOL(\mc{V}_{\Lambda})=M$. We pick such a $\mathbf{v}$ and adopt the translation $(\mathbf{v}+\Lambda)\cap\mc{S}$ as our lattice codebook for transmission. By construction, every codeword lies inside $\mc{S}$ and thus satisfies the power constraint.

    Now, by the law of large numbers, we know that with high probability, the noise vector $\mathbf{z}$ will lie inside the typical noise ball $\mc{B}(r)$ where $r=\sqrt{(1+\delta')N\eta^2}$ for a $\delta'>0$. This typical ball has the volume
    \begin{equation}
        \VOL(\mc{B}(r)) \approx \frac{\left((1+\delta')2\pi e \eta^2\right)^{\frac{N}{2}}}{\sqrt{N\pi}},
    \end{equation}
    where the approximation is due to Stirling's approximation. Since $\Lambda$ is good for channel coding, the decoding probability of error vanishes if we pick $\Lambda$ having
    \begin{equation}
        \VOL(\mc{V}_{\Lambda}) \approx \frac{\left((1+\delta)2\pi e \eta^2\right)^{\frac{N}{2}}}{\sqrt{N\pi}},
    \end{equation}
    for a $\delta>\delta'>0$. Therefore, arbitrarily reliable transmission is possible with $(\mathbf{v}+\Lambda)\cap\mc{S}$ as long as
    \begin{align}
        R &= \frac{1}{N}\log\left(M\right) \nonumber \\
        &= \frac{1}{2}\log\left(\frac{\VOL(\mc{S})}{\VOL(\mc{V}_{\Lambda})}\right)^{\frac{2}{N}} \nonumber \\
        &\approx \frac{1}{2}\log\left(\frac{(2\pi e P)/(N\pi)^{1/N}}{(1+\delta) (2\pi e\eta^2)/(N\pi)^{1/N}}\right) \nonumber \\
        &= \frac{1}{2}\log\left(\frac{P}{(1+\delta)\eta^2}\right).
    \end{align}
    Also note that this rate can be achieved by multistage decoding as $\Lambda$ is good for channel coding under multistage decoding. Letting $N\rightarrow\infty$ and $\delta\rightarrow 0$ completes the proof.
\end{IEEEproof}

\begin{remark}
    Although good lattices with sphere shaping adopted in \cite{loeliger97} and above can achieve the AWGN capacity in the asymptotically high SNR regime, it plays a little role in the recent breakthroughs of exploiting the channel structures via lattice structures. It is mainly because the spherical shaping cannot guarantee an isomorphism between messages and lattice codewords, which has been shown crucial for applications such as compute-and-forward \cite{nazer2011CF}. Later in Section~\ref{sec:AWGN}, we will discuss how to construct nested lattice codes from Construction $\pi_A$ lattices. This technique adopts nested lattice shaping and guarantees an isomoprhism.
\end{remark}
}

\section{Comparison with Construction D and Construction by Code Formula}\label{sec:compare_D}
We compare and contrast the Construction $\pi_A$ lattices and the Construction D lattices \cite{BarnesSloane83} \cite[Page 232]{conway1999sphere}. In order to make a detailed comparison, we first summarize Construction D in the following. Let $C^1 \subseteq C^2 \subseteq  \ldots \subseteq C^{L+1}$ be a sequence of nested linear codes over $\mbb{F}_p$ where $C^{L+1}$ is the trivial $(N,N)$-code and $C^l$ is a $(N,m^l)$-code for $l\in\{1,2,\ldots L\}$ with $m^1 \leq \ldots \leq m^L$. The codes are guaranteed to be nested by choosing $\{ \mathbf{g}_1,\ldots,\mathbf{g}_N \}$ which spans $C^{L+1}$ and then using the first $m^l$ vectors $\{ \mathbf{g}_1,\ldots,\mathbf{g}_{m^l} \}$ to generate $C^l$. We are now ready to state Construction D of lattices.

{\bf \underline{Construction D}} A lattice $\Lambda_{\text{D}}$ generated by Construction D with $L+1$ level is given as follows.
{\black \begin{align}
    \Lambda_{\text{D}} = \bigcup &\left\{ p^L \mbb{Z}^N + \sum_{1\leq l \leq L} p^{l-1} \sum_{1\leq i \leq m^l} a_{li} \mathbf{g}_i\right. \nonumber \\
    &\hspace{1in} \left.\vphantom{p^L \mbb{Z}^N + \sum_{1\leq l \leq L} p^{l-1} \sum_{1\leq i \leq m^l} a_{li} \mathbf{g}_i} \left|\vphantom{\sum} a_{li}\in \{0,1,\ldots, p-1\} \right. \right\},
\end{align}}
where all the operations are over $\mbb{R}^N$.

{\black Recently in \cite{YanLingWu13}, Construction D has been adopted together with nested polar codes to construct polar lattices. Polar lattice codes obtained from polar lattices with discrete Gaussian shaping is then shown to be capacity-achieving \cite{yan14polar}. Polar lattices are also shown to be able to achieve the rate distortion bound of memoryless Gaussian source in \cite{liu15polar_lossy}.} One variant of Construction D called Construction by Code Formula has attracted a lot of attention since its introduction by Forney in \cite{forney88}, see for example \cite{harshan12,KosiOggier13}. It is known that Construction by Code Formula does not always produce a lattice and it has been shown recently in \cite{KosiOggier13} that one requires the nested linear codes to be closed under Schur product in order to have a lattice. We summarize this construction in the following.

{\bf \underline{Construction by Code Formula}} A lattice $\Lambda_{\text{code}}$ generated by Construction by Code Formula over $\mbb{F}_p$ with $L+1$ levels is given as follows.
\begin{equation}
    \Lambda_{\text{code}} = C^1 + p C^2 + \ldots + p^{L-1} C^L + p^L \mbb{Z}^N. \label{eqn:code_formula}
\end{equation}

Both Construction D and Construction by Code Formula admit an efficient decoding algorithm as follows. The decoder first reduces the received signal by modulo $p\mbb{Z}$. This will get rid of all the contribution from $C^2, \ldots, C^{L+1}$ and the remainder is a codeword from the linear code $C^1$. After successfully decoding, the decoder reconstructs and subtracts out the contribution from $C^1$ and divides the results by $p$. Now the signal becomes a noisy version (with noise variance reduced by a factor of $p^2$) of a lattice point from a lattice generated by the same construction with only $L$ levels. So the decoder can then repeat the above procedure until all the codewords are decoded. In \cite{forney2000}, Forney \textit{et al.} showed that Construction D lattices together with the above decoding procedure achieves the sphere bound (Poltyrev-limit) and hence is good for channel coding.

At first glance, due to its multilevel nature, Construction $\pi_A$ looks similar to Construction D and Construction by Code Formula. Some important differences between Construction $\pi_A$ and the two constructions described above are discussed in the following.

\begin{enumerate}
    \item Construction $\pi_A$ relies solely on the ring (or $\mbb{Z}$-module) isomorphism while Construction D and Construction by Code Formula require the linear code at each level to be nested into those in the subsequent levels. The removal of such requirement makes the rate allocation and code construction much easier for the Construction $\pi_A$ lattices.
    \item {\black A fundamental difference is that Construction $\pi_A$ requires the codes used in different levels to be over different fields while Construction D allows them to be over the same field but requires the codes to be nested. Construction by code formula relaxes the nesting condition but may not always form a lattice.}
    \item The mapping from $C^1\times \ldots \times C^L$ to $\mbb{Z}/p^{L+1}\mbb{Z}$ in Construction D and Construction by Code Formula as a whole does not have the ring homomorphism property possessed by our Construction $\pi_A$. i.e., integer linear combinations of lattice points may not correspond to linear combination of codewords over $\mbb{F}_p$ for $C^1,\ldots,C^L$. The lack of ring homomorphisms renders these two constructions not straightforward to be used for applications such as compute-and-forward \cite{nazer2011CF}. However, if one does not insist on working over finite fields, Construction D can again be used for compute-and-forward. Please see the following remark.
\end{enumerate}

\begin{remark}\label{rmk:const_D_ring}
    In \cite[Proposition 2]{feng11rings}, Feng \textit{et al.} shows that Construction D can be viewed as Construction A with a code over the finite chain ring $\mbb{Z}_{p^{L-1}}$. Thus, if one would code over the ring $\mbb{Z}_{p^{L-1}}$, compute-and-forward can still be carried out in the ring level. Take \eqref{eqn:code_formula} for example, although we may not be able to know the codeword in each $C^l$, $l\in\{1,\ldots,L\}$, we will know the sum (weighted by $p^l$) as an element in the ring $\mbb{Z}_{p^{L-1}}$.
\end{remark}

\section{Construction $\pi_D$ lattices}\label{sec:const_pi_D}
Motivated by the observation made in Remark~\ref{rmk:const_D_ring}, we now consider a generalization of the Construction $\pi_A$ lattices. This generalization substantially enlarges the design space and further contains Construction D as a special case. We refer to this generalization as Construction $\pi_D$. The main enabler of this generalization is the following proposition.

\begin{proposition}\label{prop:ring_iso_pi_D}
    Let $q\in\mbb{N}$ be any natural number whose prime factorization is given by $q=\Pi_{l=1}^L p_l^{e_l}$. There exists a ring isomorphism $\mc{M}:\times_{l=1}^L \mbb{Z}_{p_l^{e_l}}\rightarrow \mbb{Z}/ q \mbb{Z}$. Moreover,
    $\sigma=\mc{M}^{-1}\circ\hspace{-3pt}\mod q\mbb{Z}$ is a ring homomorphism.
\end{proposition}
\begin{IEEEproof}
    Similar to the proof of Proposition~\ref{prop:ring_iso}.
\end{IEEEproof}

\textbf{\underline{Construction $\pi_D$}} Let $q\in\mbb{N}$ whose prime factorization is given by $q=\Pi_{l=1}^L p_l^{e_l}$. Let $m^l$, $N$ be integers such that $m^l\leq N$ and let $\mathbf{G}^l$ be a generator matrix of an $(N,m^l)$ linear code over $\mbb{Z}_{p_l^{e_l}}$ for $l\in\{1,\ldots,L\}$. Construction $\pi_D$ consists of the following steps,

\begin{enumerate}
    \item Define the discrete codebooks $C^l=\{\mathbf{x}=\mathbf{G}^l\odot\mathbf{w}^l:\mathbf{w}^l\in(\mbb{Z}_{p_l^{e_l}})^{m^l}\}$ for $l\in\{1,\ldots,L\}$.
    \item Construct $\Lambda^* \defeq \mc{M}(C^1,\ldots,C^L)$ where $\mc{M}:\times_{l=1}^L \mbb{Z}_{p_l^{e_l}}\rightarrow \mbb{Z}/ q \mbb{Z}$ is a ring isomorphism.
    \item Tile $\Lambda^*$ to the entire $\mbb{R}^N$ to form $\Lambda_{\pi_D} \triangleq  \Lambda^* + q \mbb{Z}^N$.
\end{enumerate}
Similar to $\Lambda_{\pi_A}$, it can be shown that a real vector $\boldsymbol\lambda$ belongs to $\Lambda_{\pi_D}$ if and only if $\sigma(\boldsymbol\lambda)\in C^1\times\ldots\times C^L$ where $\sigma\defeq \mc{M}^{-1}\circ\hspace{-3pt}\mod q\mbb{Z}$ is a ring homomorphism. Note that in the step 1) of the Construction $\pi_D$ procedure, we use coding over the finite chain ring $\mbb{Z}_{p_l^{e_l}}$ for the level $l$. Thus, thanks to \cite[Proposition 2]{feng11rings}, this subsumes the Construction D procedure and hence one can implement Construction D for each level. In the following, we provide an example with two levels where the first one is over a finite chain ring and the second one is over a finite field.

\begin{example}
Let us consider a two level example with $q=12 = 2^2\cdot 3$. One has $\Z/12\Z\cong \mbb{Z}_4 \times \mbb{F}_3$ and a ring isomorphism $\mc{M}(0,0)= 0$, $\mc{M}(1,1)= 1$, $\ldots$, $\mc{M}(3,2)= 11$. Let us choose
\begin{equation}
    \mathbf{G}^1 = \begin{bmatrix}
                     0 & 1 \\
                     1 & 1 \\
                   \end{bmatrix},
\end{equation}
over $\mbb{Z}_4$ and $\mathbf{G}^2=[1,1]^T$ over $\mbb{F}_3$. Also, since $\mbb{Z}_4$ is a finite chain ring, we can apply the Construction D procedure for the first level. In Fig.~\ref{fig:pi_D_lattice1}, we show the step 2) of Construction $\pi_D$ where we use the ring isomorphism $\mc{M}$ to modulate the codewords onto $\Z/12\Z$ to form $\Lambda^*$. Each element $\boldsymbol\lambda^*\in\Lambda^*$ is a coset representative of the coset $\boldsymbol\lambda^* + 12 \Z^2$. In Fig.~\ref{fig:pi_D_lattice2}, we further tile the coset representatives to the entire $\mbb{R}^2$; this corresponds to the step 3) in Construction $\pi_A$. An illustration of Construction $\pi_D$ for this particular example can be found in Fig.~\ref{fig:pi_D_lattice12}.
\begin{figure}
    \centering
    \includegraphics[width=3.5in]{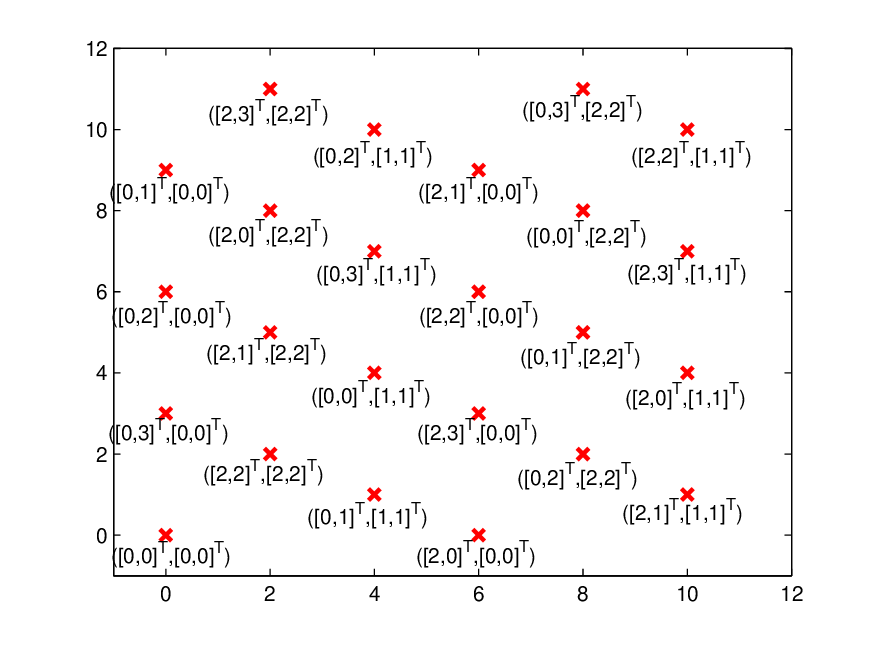}
    \caption{The set of coset representatives $\Lambda^*$ generated by $\mc{M}(C^1,C^2)$ where the corresponding codewords are also shown.}
    \label{fig:pi_D_lattice1}
\end{figure}
\begin{figure}
    \centering
    \includegraphics[width=3.5in]{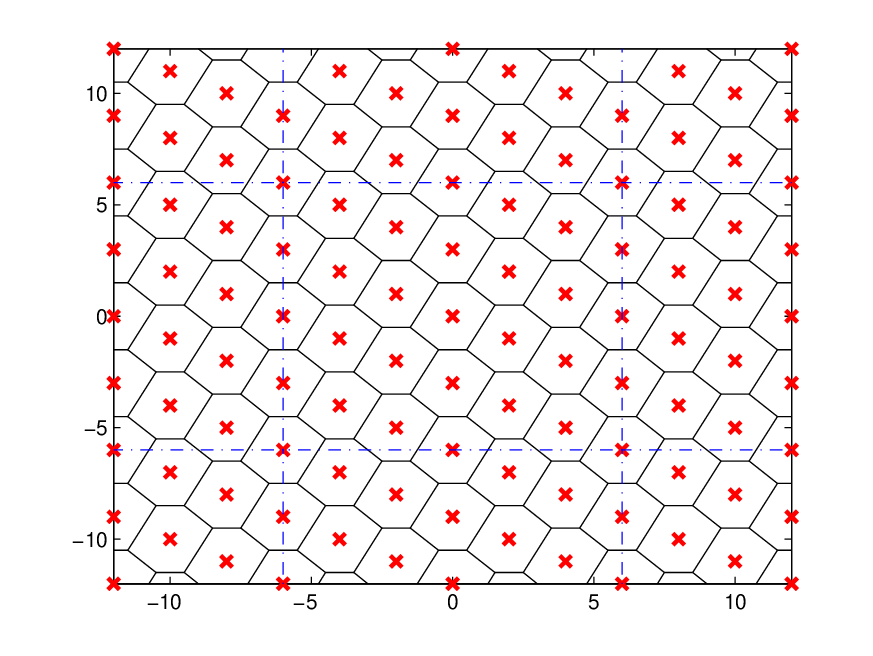}
    \caption{Tiling $\Lambda^*$ to the entire $\mbb{R}^2$ to form $\Lambda$. i.e., $\Lambda = \Lambda^*+12\Z^2$.}
    \label{fig:pi_D_lattice2}
\end{figure}
\begin{figure}
    \centering
    \includegraphics[width=2.5in]{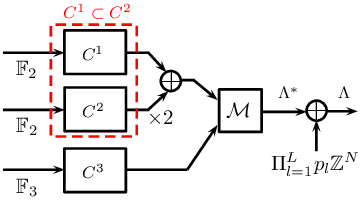}
    \caption{An example of the Construction $\pi_D$ lattice.}
    \label{fig:pi_D_lattice12}
\end{figure}
\end{example}

Note that when setting $L=1$ and $e_1=1$, Construction $\pi_D$ reduces to Construction A over a finite field $\mbb{F}_{p_1}$. Setting $L=1$ makes it Construction A over a finite chain ring $\mbb{Z}_{p_1^{e_1}}$, which subsumes Construction D as a special case. Finally, when setting $e_1=\ldots=e_L=1$, we obtain Construction $\pi_A$. Hence, Construction $\pi_D$ is a general means of constructing lattices from codes and contains Construction A, Construction D, and Construction $\pi_A$ as special cases. Moreover, Construction $\pi_D$ is more than these three special cases. Particularly, $q$ can take any natural number \textit{regardless} its prime factorization. Thus,  Construction $\pi_D$ substantially expands the design space and further eases the rate allocation problem.

To show the ability to produce Poltyrev good lattices, for Construction $\pi_D$, one can follow the proof in Theorem~\ref{thm:goodness} with a careful treatment to those levels with $e_l\neq 1$. One option is to use Construction D for those levels, i.e., one uses a sequence of $e_l$ nested linear codes to construct a linear code over $\mbb{Z}_{p_l^{e_l}}$. Another option is to adopt a capacity-achieving linear code over $\mbb{Z}_{p_l^{e_l}}$ proposed in \cite{HitronErez12} at the $l$th level. 

\section{Discussions}\label{sec:discussion}
In this section, we first point out that ring isomorphisms are \textit{not necessary} and $\mbb{Z}$-module isomorphisms suffice in order to get a lattice. This is practically relevant as this increases the design space. An easy way to generate a $\mbb{Z}$-module isomorphism is also provided which closely follows the set partition rule of Ungerboeck \cite{Ungerboeck82}. We then provide a brief comparison of decoding complexity between the Construction $\pi_A$ lattices and the Construction A lattices with codes over prime fields.

\subsection{$\mbb{Z}$-Module Isomorphisms Suffice}
One may have already noticed that in Proposition~\ref{prop:lattice} and Theorem~\ref{thm:goodness}, we only use the fact that $\mc{M}$ is a $\mbb{Z}$-module isomorphism instead of a ring isomorphism. In fact, since a lattice is a free $\mbb{Z}$-module so the requirement of ring isomorphisms may be too strong and $\mbb{Z}$-module isomorphisms suffice. However, the requirement of ring isomorphisms appears to be imperative for some applications such as compute-and-forward \cite{huangisit14}. In the sequel, we discuss Construction $\pi_A$ with $\mbb{Z}$-module isomorphisms.

We begin by noting that if we regard the both sides of \eqref{eqn:isomorphic} as finitely-generated Abelian groups, i.e., $\mbb{Z}$-modules, one has that the following $\mbb{Z}$-module homomorphisms exists
\begin{equation}\label{eqn:module_homo_prod}
    \varphi: \mbb{Z} \overset{\hspace{-3pt}\mod \Pi_{l=1}^L p_l\mbb{Z}}{\rightarrow} \mbb{Z}/\Pi_{l=1}^L p_l\mbb{Z} \overset{\mc{M}^{-1}}{\rightarrow} \mbb{Z}/p_1\mbb{Z}\times\ldots\times\mbb{Z}/p_L\mbb{Z},
\end{equation}
where now $\mc{M}$ is a $\mbb{Z}$-module isomorphism. One can show that Construction $\pi_A$ with this $\mc{M}$ in the step 2) would result in a lattice $\Lambda$ and $\boldsymbol\lambda\in\Lambda$ if and only if $\varphi(\boldsymbol\lambda)\in C^1\times\ldots\times C^L$. In the following, we provide an explicit construction of a $\mbb{Z}$-module isomorphism and give an example in $\Zw$ which relates the proposed multilevel lattices to the Ungerboeck set partitions \cite{Ungerboeck82}.

\begin{theorem}
    Let $p_1,\ldots,p_L$ be a collection of primes which are relatively prime. The following mapping
    \begin{equation}\label{eqn:general_homo}
        \mc{M}(v^1,\ldots,v^L) \defeq \sum_{l=1}^L v^l\Pi_{l'=1,l'\neq l}^L p_{l'} \hspace{-3pt}\mod \Pi_{l=1}^L p_l\mbb{Z},
    \end{equation}
    where $v^l\in\mbb{F}_{p_l}$, is a $\mbb{Z}$-module isomorphism from $\times_{l=1}^L \mbb{F}_{p_l}$ to $\mbb{Z}/\Pi_{l=1}^L p_l\mbb{Z}$. Therefore, $\varphi\defeq \mc{M}^{-1}\circ \hspace{-3pt}\mod\Pi_{l=1}^L\Zw$ is a $\mbb{Z}$-module homomorphism.
\end{theorem}
\begin{IEEEproof}
    Let $v^l_k\in \mbb{F}_{p_l}$ for $k\in\{1,2\}$ and $l\in\{1,\ldots,L\}$. Consider
    \begin{equation}
        \mc{M}(v^1_k,\ldots,v^L_k) = \sum_{l=1}^L v^l_k\Pi_{l'=1,l'\neq l}^L p_{l'} \hspace{-3pt}\mod\Pi_{l=1}^L p_l\mbb{Z}.
    \end{equation}
    One has that
    \begin{align}
        &\mc{M}(v^1_1,\ldots,v^L_1) + \mc{M}(v^1_2,\ldots,v^L_2) \hspace{-3pt}\mod\Pi_{l=1}^L p_l\mbb{Z} \nonumber \\
        &= \sum_{l=1}^L (v^l_1+v^l_2)\Pi_{l'=1,l'\neq l}^L p_{l'} \hspace{-3pt}\mod\Pi_{l=1}^L p_l\mbb{Z} \nonumber \\
        &= \sum_{l=1}^L (v^l_1\oplus v^l_2+\boldsymbol\zeta_l p_l)\Pi_{l'=1,l'\neq l}^L p_{l'} \hspace{-3pt}\mod\Pi_{l=1}^L p_l\mbb{Z} \nonumber \\
        &= \sum_{l=1}^L (v^l_1\oplus v^l_2)\Pi_{l'=1,l'\neq l}^L p_{l'} + \sum_{l=1}^L \boldsymbol\zeta_l \Pi_{l=1}^L p_{l} \hspace{-3pt}\mod\Pi_{l=1}^L p_l\mbb{Z} \nonumber \\
        &= \sum_{l=1}^L (v^l_1\oplus v^l_2)\Pi_{l'=1,l'\neq l}^L \phi_{l'} \hspace{-3pt}\mod\Pi_{l=1}^L p_l\mbb{Z} \nonumber\\
        &= \mc{M}(v^1_1\oplus v^1_2,\ldots,v^L_1\oplus v^L_2),
    \end{align}
    where $\boldsymbol\zeta_l \in \mbb{Z}^N$.
\end{IEEEproof}
It should be noted that there exist many other $\mbb{Z}$-module homomorphism and the design space is quite large. We now provide an example in $\Zw$ and relate the above construction of $\mbb{Z}$-module isomorphism to the Ungerboeck set partitions \cite{Ungerboeck82}.

\begin{example}\label{exp:40pt_z_module}
    Consider $\Zw$ the ring of Eisenstein integers. Let $\phi_1=3+2\omega$ and $\phi_2=1-2\omega$. One can verify that both $\phi_1$ and $\phi_2$ are Eisenstein primes with $|\phi_1|^2=|\phi_2|^2=7$ and $\phi_1$ and $\phi_2$ are relatively prime. Thus, we have $\Zw/\phi_1 \phi_2\Zw\cong \mbb{F}_7\times\mbb{F}_7$. The above algorithm would produce a $\mbb{Z}$-module isomorphism given by
    \begin{equation}\label{eqn:module_homo}
        \mc{M}(v^1,v^2) \triangleq \phi_2 v^1 + \phi_1 v^2 \hspace{-3pt}\mod \phi_1 \phi_2\Zw.
    \end{equation}
    where $v^1, v^2\in\mbb{F}_q$. This isomorphism is shown in Fig.~\ref{fig:z_module_49} where the first and second digits represent elements in the first and second fields, respectively. One observes that this mapping closely follows the set partition rules of Ungerboeck that the minimum intra-subset distance is maximized when partitioning at each level.
    \begin{figure}
    \centering
    \includegraphics[width=3.5in]{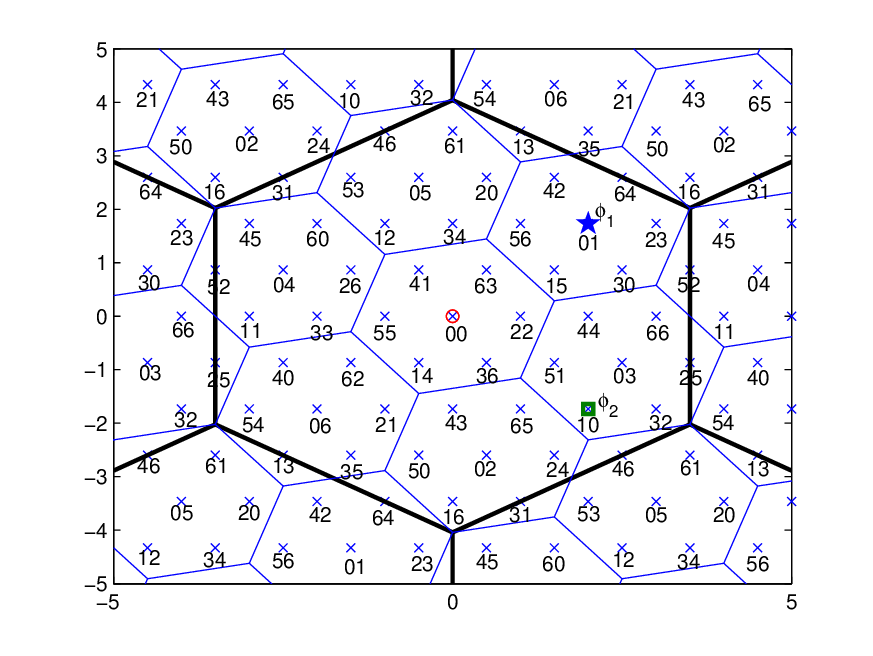}
    \caption{An example of the $\mbb{Z}$-module isomorphism in \eqref{eqn:general_homo} from $\Zw$ with $\phi_1 = 3+2\omega$ and $\phi_2 = 1-2\omega$.}
    \label{fig:z_module_49}
    \end{figure}
\end{example}

\subsection{Comparison of Complexity}
To emphasize the advantage of Construction $\pi_A$ lattices over Construction A lattices, in Fig.~\ref{fig:complexity}, we present a rough comparison of the decoding complexity between lattices from these two constructions. The underlying linear codes are chosen to be non-binary LDPC codes. Recall that for Construction A lattices the decoding complexity is dominated by $|\Lambda^*|$ while for the Construction $\pi_A$ lattices over $\mbb{Z}$, it only depends on the greatest prime divisor of $|\Lambda^*|$. For coding over $\mbb{F}_p$, we assume that a $p$-ary LDPC code is implemented for which the decoding complexity is reported to be roughly $O(p\log(p))$ \cite{DaveyMackay}. Note that for Construction $\pi_A$, we exclude those lattices that can also be generated by Construction A and those lattices that would result in higher complexity than Construction A. Because for those parameters, one could just use Construction A. One observes in Fig.~\ref{fig:complexity} that Construction $\pi_A$ significantly reduces the decoding complexity. Moreover, one can expect the gain to be larger as the constellation size $|\Lambda^*|$ increases.

The same comparison is also performed for lattices over $\Zi$ and over $\Zw$. Note that allowing lattices over such rings of integers enlarges the design space and may further decrease the decoding complexity. For example, $|\Lambda^*|=25$ was not an option for Construction $\pi_A$ over $\mbb{Z}$; however, we know that $5\Zi$ splits into two prime ideals $(1+2i)\Zi$ and $(1-2i)\Zi$. Moreover, these two prime ideals are relatively prime so the CRT gives
\begin{align}
    \Zi/5\Zi &\cong \Zi/(1+2i)\Zi \times \Zi/(1-2i)\Zi \nonumber \\
    &\cong \mbb{F}_5\times \mbb{F}_5.
\end{align}
One can use Construction $\pi_A$ over $\Zi$ with these two prime ideals. The resulted lattice would have decoding complexity dominated by coding over $\mbb{F}_5$. Another example can be found when $|\Lambda^*|=49$ which was not an option for Construction $\pi_A$ over $\mbb{Z}$. However, $7\Zw = (2+3\omega)(-1-3\omega)\Zw$. Hence, the CRT gives
\begin{align}
    \Zw/7\Zw &\cong \Zw/(2+3\omega)\Zw \times \Zw/(-1-3\omega)\Zw \nonumber \\
    &\cong \mbb{F}_7\times \mbb{F}_7.
\end{align}
One can use Construction $\pi_A$ over $\Zw$ with prime ideals $(2+3\omega)\Zw$ and $(-1-3\omega)\Zw$. The {\black resulting} lattice would have decoding complexity dominated by coding over $\mbb{F}_7$.
\begin{figure}
    \centering
    \includegraphics[width=3.5in]{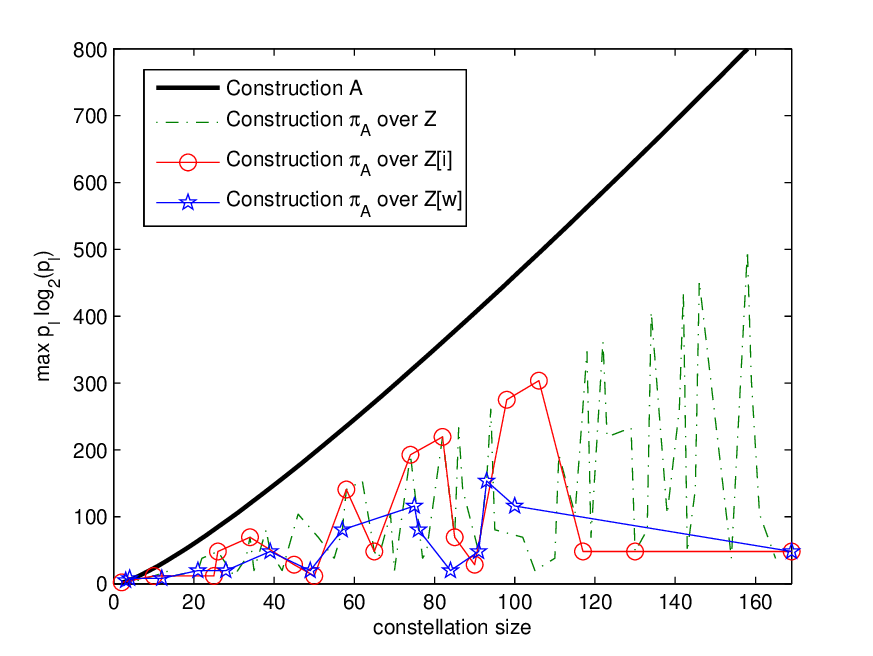}
    \caption{A rough comparison of decoding complexity for Construction A and Construction $\pi_A$ lattices.}
    \label{fig:complexity}
\end{figure}

\section{Low-Complexity Decoders}\label{sec:para_dec}
In Section~\ref{sec:prod_const}, it has been shown that Construction $\pi_A$ is able to produce a sequence of lattices that is Poltyrev-good under multistage decoding. We now propose two decoders which further take advantage of the additional structure of Construction $\pi_A$ lattices. Before starting, we note that the algorithms proposed here can be straightforwardly extended to Construction $\pi_D$ lattices but we only present the algorithms for Construction $\pi_A$ lattices for the sake of brevity. A remark (Remark~\ref{rmk:pi_D_decoding}) will be given later to discuss the extension to Construction $\pi_D$ lattices.

The key property that we exploit here is that from CRT, any $a\in\mbb{Z}$ can be uniquely represent as $a=\mc{M}(b^1,\ldots,b^L)+ \tilde{a}\cdot \Pi_{l=1}^L p_l$ where $b^l\in\mbb{F}_{p_l}$ and $\tilde{a}\in \mbb{Z}$ and
\begin{equation}\label{eqn:unique_pl}
    a\hspace{-3pt}\mod p_l = b^l.
\end{equation}
We now discuss the first proposed decoder which is referred to as the serial modulo decoder (SMD). This decoder is motivated by a decoding algorithm of Construction D lattices \cite{forney2000} and heavily relies on the additional structure \eqref{eqn:unique_pl} provided by CRT . The SMD first removes the contribution from all but the first level by performing $\hspace{-3pt}\mod p_1\mbb{Z}^N$ to form
\begin{align}\label{eqn:sub_dec_1}
    \mathbf{y}^1 &\defeq \mathbf{y} \hspace{-3pt}\mod p_1\mbb{Z}^N = \left(\mathbf{x} + \mathbf{z}\right) \hspace{-3pt}\mod p_1\mbb{Z}^N \nonumber \\
    &= \left( \mc{M}(\mathbf{c}^1,\ldots,\mathbf{c}^L) + \Pi_{l=1}^L p_l \boldsymbol\zeta + \mathbf{z}\right) \hspace{-3pt}\mod p_1\mbb{Z}^N \nonumber \\
    &\overset{(a)}{=} \left(\mathbf{c}^1 + \mathbf{z}\hspace{-3pt}\mod p_1\mbb{Z}^N\right)\hspace{-3pt}\mod p_1\mbb{Z}^N,
\end{align}
where (a) follows from the distributive property of the modulo operation and \eqref{eqn:unique_pl}. This procedure transforms the channel into a single level additive $\hspace{-3pt}\mod p_1\mbb{Z}^N$ channel. The decoder then forms $\mathbf{\hat{c}}^1$ an estimate of $\mathbf{c}^1$ from $\mathbf{y}^1$ by decoding the linear code $C^1$. {\black This transformation converts the AWGN channel into the $\hspace{-3pt}\mod p_1$ channel and thus is suboptimal; however, the loss is negligible in the high  SNR regime as mentioned in \cite{forney2000}.}

For the levels $s\in\{2,\ldots,L\}$, the decoder assumes all the previous levels are correctly decoded, i.e., $\mathbf{\hat{c}}^l=\mathbf{c}^l$ for $l<s$. It then subtracts all the contributions from the previously decoded levels from $\mathbf{y}$ to form
\begin{equation}\label{eqn:sub_after_s}
    \mc{M}(\mathbf{0},\ldots,\mathbf{0},\mathbf{c}^s,\ldots,\mathbf{c}^L) + \Pi_{l=1}^L p_l \boldsymbol\zeta + \mathbf{z}.
\end{equation}
Note that both $\mc{M}(\mathbf{0},\ldots,\mathbf{0},\mathbf{c}^s,\ldots,\mathbf{c}^L)$ and $\Pi_{l=1}^L p_l \boldsymbol\zeta$ are multiples of $\Pi_{l=1}^{s-1} p_l$ and dividing \eqref{eqn:sub_after_s} by $\Pi_{l=1}^{s-1} p_l$ results in
\begin{equation}
    \mc{M}^s(\mathbf{c}^s,\ldots,\mathbf{c}^L) + \Pi_{l=s}^L p_l \boldsymbol\zeta + \mathbf{\tilde{z}}^s,
\end{equation}
where $\mc{M}^s$ is a bijective mapping from $\mbb{F}_{p_s}\times\ldots\times \mbb{F}_{p_L}$ to $\mbb{Z}/\Pi_{l=s}^L p_l \mbb{Z}$ and $\mathbf{\tilde{z}}^s\defeq \mathbf{z}/\Pi_{l=1}^{s-1} p_l$ whose elements are i.i.d. Gaussian distributed with zero mean and variance $\eta^2/(\Pi_{l=1}^{s-1} p_l)^2$. We can now again remove the contributions from the next levels to form
\begin{align}\label{eqn:sub_dec_s}
    \mathbf{\tilde{y}}^s &= \left(\mc{M}^s(\mathbf{c}^s,\ldots,\mathbf{c}^L) + \Pi_{l=s}^L p_l \boldsymbol\zeta + \mathbf{\tilde{z}}^s\right) \hspace{-3pt}\mod p_s\mbb{Z}^N \nonumber \\
    &\overset{(a)}{=} \left( d_s\odot\mathbf{c}^s+ \mathbf{\tilde{z}}^s\hspace{-3pt}\mod p_s\mbb{Z}^N\right) \hspace{-3pt}\mod p_s\mbb{Z}^N,
\end{align}
where (a) follows by the structure of mapping in \eqref{eqn:ring_iso} and $d_s = (a_s q_s/\Pi_{l=1}^{s-1}p_l)\hspace{-3pt}\mod p_s$. Note that since $C^s$ is linear, $d_s\odot \mathbf{c}_s\in C^s$. This procedure makes the channel experienced by the $s$th coded stream a single level additive $\hspace{-3pt}\mod p_s\mbb{Z}^N$ channel with noise variance reduced by a factor of $(\Pi_{l=1}^{s-1}p_{l})^2$. The decoder then forms $\mathbf{\hat{c}}^s$ an estimate of $\mathbf{c}^s$ from $\mathbf{\tilde{y}}^s$ by decoding the linear code $C^{s}$.

In the last level of decoding, one does not have to perform the modulo operation as there is only one level left. Therefore, the decoder at the last level directly decodes the uncoded integer $\boldsymbol\zeta$ by quantizing $\mathbf{\tilde{y}}^{L+1} \defeq \boldsymbol\zeta + \mathbf{\tilde{z}}^{L+1}$ to the nearest integer vector. We summarize the decoding procedure of the proposed SMD in Fig.~\ref{fig:sub_dec}.
\begin{figure}
    \centering
    \includegraphics[width=2.8in]{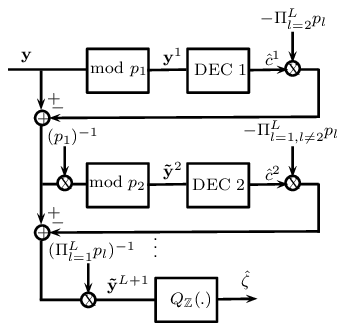}
    \caption{The proposed SMD decoder for Construction $\pi_A$ lattices.}
    \label{fig:sub_dec}
\end{figure}

We now propose another decoder which is very similar to the SMD but can be implemented in a parallel fashion. Thus, this decoder is referred to as the parallel modulo decoder (PMD). Due to its parallel nature, depending on the total number of levels $L$, this decoder can have substantially smaller latency than the multistage decoding and the SMD.

Note that, from \eqref{eqn:unique_pl}, $\mathbf{x}\hspace{-3pt}\mod p_l\mbb{Z}^N = \mathbf{c}^l$ for every $l\in\{1,\ldots,L\}$. For the PMD, we simultaneously form
\begin{align}\label{eqn:para_dec_s}
    \mathbf{y}^s &= \left(\mathbf{x} + \mathbf{z}\right) \hspace{-3pt}\mod p_s\mbb{Z}^N \nonumber \\
    &= \left( \mc{M}(\mathbf{c}^1,\ldots,\mathbf{c}^L) + \Pi_{l=1}^L p_l \boldsymbol\zeta + \mathbf{z}\right) \hspace{-3pt}\mod p_s\mbb{Z}^N \nonumber \\
    &\overset{(a)}{=} \left(\mathbf{c}^s + \mathbf{z}\hspace{-3pt}\mod p_s\mbb{Z}^N\right)\hspace{-3pt}\mod p_s\mbb{Z}^N,
\end{align}
for $s\in\{1\ldots,L\}$ where $(a)$ follows again from \eqref{eqn:unique_pl}. The decoder then directly forms $\hat{\mathbf{c}}^s$ an estimate of $\mathbf{c}^s$ from $\mathbf{y}^s$ by decoding the linear code $C^s$ for $s\in\{1,\ldots,L\}$. Now, instead of having a reduced noise $\mathbf{\tilde{z}}^s$ at the $s$th level as in \eqref{eqn:sub_dec_s}, the noise random variables before the modulo operations are the same for all the levels. Thus, the performance of the PMD would be worse than that of the SMD for a same Construction $\pi_A$ lattice. For the last step, the parallel decoder finds the uncoded integer $\boldsymbol\zeta$ from $\mathbf{y}^{L+1}\defeq \mathbf{\tilde{y}}^{L+1}$.

\begin{figure}
    \centering
    \includegraphics[width=3.2in]{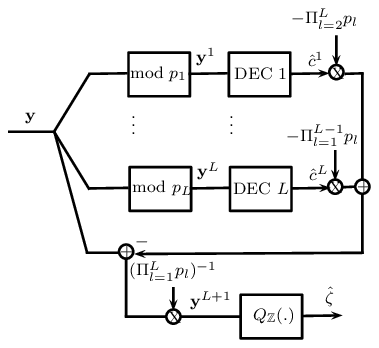}
    \caption{The proposed PMD decoder for Construction $\pi_A$ lattices.}
    \label{fig:para_dec}
\end{figure}

\begin{remark}
    Depending on the performance and latency requirements, one can implement a mixed decoder which is in between SMD and PMD in a fashion that some levels are implemented serial and others are implemented parallel.
\end{remark}

\begin{remark}\label{rmk:pi_D_decoding}
    Note that the proposed algorithms can be extended to decoding of lattices from Construction $\pi_D$ by first performing $\hspace{-3pt}\mod p_l^{e_l}\mbb{Z}$ at the $l$th level to get a noisy version of the $l$th codeword which is over the finite chain ring $\mbb{Z}_{p_l^{e_l}}$. Then the decoding problem becomes decoding of a Construction D lattice constructed over $\mbb{F}_{p_l}$ with $e_l$ levels and the decoding algorithm described in Section~\ref{sec:compare_D} can be used.
\end{remark}

We present some numerical results which consider using Construction $\pi_A$ lattices with the hypercube shaping over the AWGN channel. i.e., we consider the AWGN channel given by $\mathbf{y}' = \gamma \mathbf{x} + \mathbf{z}'$ where $\mathbf{x}$ is an element of $\Lambda\hspace{-3pt}\mod \Pi_{l=1}^L p_l \mbb{Z}^N$ a Construction $\pi_A$ lattice shaped by a hypercubic coarse lattice, $\gamma$ is for the power constraint, and $\mathbf{z}'$ is the additive Gaussian noise having distribution $\mc{N}(\mathbf{0},\mathbf{I})$. We can equivalently consider the model
\begin{equation}
    \mathbf{y} = \mathbf{x} + \mathbf{z},
\end{equation}
where $\mathbf{z}\defeq \mathbf{z'}/\gamma$ having covariance matrix $\mathbf{I}/\gamma^2$.

We now discuss the information rates achievable by different decoders. For the multistage decoder, one has
\begin{align}
    R_{\text{MSD}} &= I(\msf{X};\msf{Y}) \nonumber \\
        &\overset{(a)}{=} I(\msf{C}^1,\ldots,\msf{C}^L;\msf{Y}) \nonumber \\
        &\overset{(b)}{=} I(\msf{C}^1;\msf{Y})+\sum_{l=2}^L I(\msf{C}^l;\msf{Y}|\msf{C}^1,\ldots,\msf{C}^{l-1}),
\end{align}
where (a) is due to the fact that $\mc{M}$ is a ring isomorphism and hence is bijective and (b) follows from the chain rule of mutual information \cite{cover91}. The achievable information rates for the SMD and PMD can be analyzed similarly and are given by
\begin{equation}
    R_{\text{SMD}} = I(\msf{C}^1;\msf{Y}^1) +\sum_{s=2}^L I(\msf{C}^s;\msf{\tilde{Y}}^s|\msf{C}^{1},\ldots,\msf{C}^{s-1}),
\end{equation}
and
\begin{equation}
    R_{\text{PMD}} = I(\msf{C}^1;\msf{Y}^1) +\sum_{s=2}^L I(\msf{C}^s;\msf{Y}^s),
\end{equation}
respectively.

The information rates achievable by the multistage decoder, SMD, and PMD are computed via Monte-Carlo simulation. In Fig.~\ref{fig:sub_decoder}, we provide two examples with two levels where the lattices are generated by $p_1=2$, $p_2 = 3$ and $p_1=2$, $p_2 = 13$, respectively. One observes that for both cases, as expected, the multistage decoder performs the best among these decoders as it is also the most complex one. Also, since we are using the hypercube shaping, the coding scheme suffers from a loss of 1.53 dB in the high SNR regime that corresponds to the shaping gain. On the other hand, although being suboptimal, the SMD can support information rates close to that provided by the multistage decoder, especially in the medium and high SNR regimes. For the PMD, the achievable rates are much worse than the other two in the low SNR regime but it is still of interest in the high SNR regime due to its low complexity and low latency.

\begin{figure}
    \centering
    \includegraphics[width=3.5in]{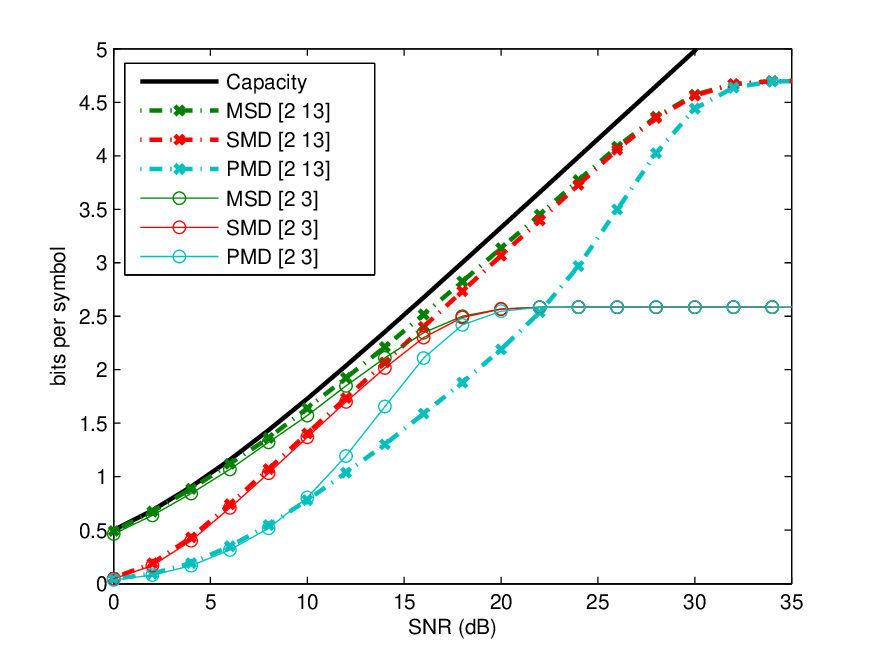}
    \caption{Average achievable rates for constellations with different size.}
    \label{fig:sub_decoder}
\end{figure}


\section{Nested Lattice Codes from Construction $\pi_A$}\label{sec:AWGN}

In this section, we construct multilevel nested lattice codes from Construction $\pi_A$ lattices. Our construction closely follow the one by Ordentlich and Erez \cite{ordentlich_erez_simple} rather than the frequently used one by Erez and Zamir in \cite{erez04}. For codes from the proposed construction, an isomorphism between lattice codewords and messages can be easily identify as detailed in \cite{huangisit14}. Again, we only consider constructing nested lattice codes over $\mbb{Z}$ but the generalizations to $\Zi$ and $\Zw$ are straightforward.

\subsection{Construction and Main Result}
Let $p_1,\ldots,p_L$ be distinct primes and $\mc{M}:\times_{l=1}^L\mbb{F}_{p_l}\rightarrow \mbb{Z}/\Pi_{l=1}^L p_l\mbb{Z}$ be a ring isomorphism. We first generate a pair of nested linear codes $(C^l_f,C^l_c)$ such that $C^l_c\subseteq C^l_f$ for each $l\in\{1,\ldots,L\}$ as follows,
\begin{align}
    C^l_c &= \{\mathbf{G}^l_c\odot\mathbf{w}^l | \mathbf{w}^l\in\mbb{F}_{p_l}^{m^l_c} \}, \\
    C^l_f &= \{\mathbf{G}^l_f\odot\mathbf{w}^l | \mathbf{w}^l\in\mbb{F}_{p_l}^{m^l_f} \},
\end{align}
where $\mathbf{G}^l_c$ is a $N\times m^l_c$ matrix and
\begin{equation}
    \mathbf{G}^l_f = \begin{bmatrix}
                       \mathbf{G}^l_c & \mathbf{\tilde{G}}^l \\
                     \end{bmatrix},
\end{equation}
where $\mathbf{\tilde{G}}^l$ is a $N\times (m^l_f-m^l_c)$ matrix. We then generate (scaled) lattices $\Lambda_f$ and $\Lambda_c$ from Construction $\pi_A$ with the linear codes $C^l_f$ and $C^l_c$, respectively, as follows.
\begin{align}
    \Lambda_f &\triangleq  \gamma\left(\Pi_{l=1}^L p_l\right)^{-1} \mc{M}(C^1_f,\ldots,C^L_f) + \gamma\mbb{Z}^N, \nonumber \\
    \Lambda_c &\triangleq  \gamma\left(\Pi_{l=1}^L p_l\right)^{-1} \mc{M}(C^1_c,\ldots,C^L_c) + \gamma\mbb{Z}^N,
\end{align}
where $\gamma$ is chosen such that the $\sigma^2(\Lambda_c)=P$. Clearly, $\Lambda_c\subseteq\Lambda_f$ and the design rate is given by
\begin{equation}
    R_{\text{design}} = \sum_{l=1}^L\frac{m^l_f-m^l_c}{N}\log(p_l).
\end{equation}
The design rate becomes the actual rate if every $\mathbf{G}^l_f$ is full-rank which will be fulfilled with high probability.


\subsection{Encoding and Decoding}
The transmitter first decomposes its message into $(\mathbf{w}^1,\ldots,\mathbf{w}^L)$, where $\mathbf{w}^l$ is a length $(m_f^l-m_c^l)$ vector over $\mbb{F}_{p_l}$, and bijectively maps it to a lattice point $\mathbf{t}\in\Lambda_f\cap\mc{V}_{\Lambda_c}$ where
\begin{equation}
    \mathbf{t} = \left(\gamma\left(\Pi_{l=1}^L p_l\right)^{-1}\mc{M}(\mathbf{c}^1,\ldots,\mathbf{c}^L) + \gamma \boldsymbol\zeta\right) \hspace{-3pt}\mod\Lambda_c,
\end{equation}
with $\boldsymbol\zeta\in\mbb{Z}^N$ and $\mathbf{c}^l\defeq\mathbf{G}_f^l\odot[\mathbf{0}_{m_c^l}~\mathbf{w}^l]^T$. It then sends a dithered version
\begin{equation}
    \mathbf{x} = (\mathbf{t}-\mathbf{u}) \hspace{-3pt}\mod\Lambda_c.
\end{equation}

Upon receiving $\mathbf{y}$, the receiver scales it by the linear MMSE estimator given by
\begin{equation}
    \alpha \defeq \frac{P}{P+\eta^2},
\end{equation}
and adds the dithers back to form
\begin{align}
    &[\alpha\mathbf{y} + \mathbf{u}] \hspace{-3pt}\mod\Lambda_c = [\mathbf{t}-(1-\alpha)\mathbf{x}+\alpha\mathbf{z}] \hspace{-3pt}\mod\Lambda_c \nonumber \\
    &= [\mathbf{t}+\mathbf{z}_{eq}] \hspace{-3pt}\mod\Lambda_c \nonumber \\
    &= \left[\gamma\left(\Pi_{l=1}^L p_l\right)^{-1} \mc{M}(\mathbf{c}^1_f,\ldots,\mathbf{c}^L_f) + \gamma \boldsymbol\zeta +\mathbf{z}_{eq}\right] \hspace{-3pt}\mod\Lambda_c
\end{align}
where
\begin{equation}
    \mathbf{z}_{eq} \defeq  \alpha\mathbf{z}-(1-\alpha)\mathbf{x}\hspace{-3pt}\mod \Lambda_c,
\end{equation}
with
\begin{align}
    \frac{1}{N}\mbb{E}\|\msf{Z}_{eq}\|^2 &\leq \frac{1}{N}\mbb{E}\|\alpha\msf{Z}-(1-\alpha)\msf{X}\|^2 \nonumber \\
    &= (1-\alpha)^2 P + \alpha^2\eta^2\nonumber \\
    &= \frac{P\eta^2}{P+\eta^2}.
\end{align}
{\black Due to the random dither, $\mathbf{t}$ and $\mathbf{z}_{eq}$ are statistically independent to each other. One can now perform multistage decoding to decode the fine lattice point $\mathbf{t}$ by decoding the equivalent codewords $\mathbf{c}^l$ for $l\in\{1,\ldots,L\}$ level by level.} 

\subsection{Achievable Rate}
Let $\msf{Z}^*_{eq}$ be the i.i.d. Gaussian random vector having distribution $\mc{N}(0,\sigma^2_{eq})$ where $\sigma^2_{eq}\defeq \frac{P\eta^2}{P+\eta^2}$.
{\black The achievable rate of the proposed nested lattice codes is given in the following theorem.
\begin{theorem}
    For the AWGN channel, there exists a sequence of nested lattice codes from the proposed ensemble that can achieve the following rate under multistage decoding,
    \begin{equation}
        R = \frac{1}{2}\log\left(1+\frac{P}{\eta^2}\right) -\frac{1}{2}\log(2\pi e G(\Lambda_c)) + \frac{1}{N} D(\msf{Z}_{eq}||\msf{Z}^*_{eq}),
    \end{equation}
    where $D(.|.)$ is the Kullback-Leibler divergence \cite{cover91}.
\end{theorem}
Before proving this theorem, we discuss the implications of this result. We first note that if $\Lambda_c$ happens to be good for MSE quantization, then $G(\Lambda_c)\rightarrow 1/2\pi e$ and
\begin{equation}
    \frac{1}{N}D(\msf{Z}_{eq}\|\msf{Z}^*_{eq})\rightarrow 0,
\end{equation}
from the main result in \cite{zamir96}. This will imply the existence of capacity-achieving multilevel nested lattice codes under multistage decoding. Unfortunately, we have not been able to prove the existence of such $\Lambda_c$ with our construction\footnote{{\black Our preliminary result in \cite{huangisit14} falsely claims that we can prove the existence of such lattices with our construction. The proof there was wrong mainly because the correlation between codewords induced by the proposed construction prevents direct usage of the arguments in \cite{ordentlich_erez_simple}.}}. On the other extreme, if $\Lambda_c = \gamma \mbb{Z}^N$, that is, hypercube shaping, $G(\Lambda_c)=1/12$, we have
\begin{align}
    R &= \frac{1}{2}\log\left(1+\frac{P}{\eta^2}\right) -\frac{1}{2}\log(\frac{\pi e}{6}) + \frac{1}{N} D(\msf{Z}_{eq}||\msf{Z}^*_{eq}) \nonumber \\
    &\rightarrow\frac{1}{2}\log\left(1+\frac{P}{\eta^2}\right) -\frac{1}{2}\log(\frac{\pi e}{6}),
\end{align}
in the limit as SNR tends to infinity. This can be justified by observing that $\alpha\rightarrow 1$ and thus $\msf{Z}_{eq}\rightarrow\msf{Z}^*_{eq}$ as SNR$\rightarrow\infty$. This result conforms with the 1.53 dB loss in shaping gain in the asymptotically high SNR regime \cite{forney2000}.
}

\begin{IEEEproof}
{\black Lemma~\ref{lma:Vol_fine} in Appendix~\ref{apx:Vol_fine} establishes that there exists a sequence of the proposed lattices whose probability of error under multistage decoding can be made arbitrarily small as $N\rightarrow \infty$ if
\begin{equation}
    \text{Vol}(\Lambda_f)^{\frac{2}{N}} > 2\pi e \sigma_{eq}^2 2^{-\frac{2}{N}D(\msf{Z}_{eq}\|\msf{Z}^*_{eq})}.
\end{equation}
Therefore, there exists a sequence of proposed nested lattice codes with hypercube shaping that can achieve the design rate per real dimension given by
}
\begin{align}
    &R_{\text{design}} = \frac{1}{N}\log\left(\frac{\text{Vol}(\Lambda_c)}{\text{Vol}(\Lambda_f)}\right) \nonumber \\
    &= \frac{1}{N}\log (\text{Vol}(\Lambda_c)) - \frac{1}{N}\log(\text{Vol}(\Lambda_f)) \nonumber \\
    &\overset{N\rightarrow\infty}{\rightarrow} \frac{1}{2}\log\frac{P}{G(\Lambda_c)}-\frac{1}{2}\log2\pi e\sigma_{eq}^2 2^{-\frac{2}{N}D(\msf{Z}_{eq}\|\msf{Z}^*_{eq})} \nonumber \\
    &= \frac{1}{2}\log\left(1+\frac{P}{\eta^2}\right) -\frac{1}{2}\log(2\pi e G(\Lambda_c)) +\frac{1}{N} D(\msf{Z}_{eq}||\msf{Z}^*_{eq}).
\end{align}
Moreover, as mentioned above, with high probability, each $\mathbf{G}_f^l$ is full rank and the design rate becomes the actual rate.
\end{IEEEproof}

\section{Conclusions}\label{sec:conclusion}
In this paper, a subclass of Construction A lattices called Construction $\pi_A$ has been studied. This construction has been shown to be able to produce a sequence of lattices that is good for channel coding under multistage decoding. Inspired by the efficient decoding algorithm for Construction D lattices, two low-complexity decoding algorithms have been proposed and shown offering reasonably good performance in the medium and high SNR regimes. {\black As an important application, Construction $\pi_A$ lattices have been used to construct nested lattice code ensemble that guarantees an isomorphism between lattice codewords and messages. The achievable rate of the proposed multilevel nested lattice codes under multistage decoding has then been analyzed.} A generalization of Construction $\pi_A$ called Construction $\pi_D$ was also studied which substantially enlarges the design space and subsumes Construction A with codes over prime fields, Construction D, and Construction $\pi_A$ as special cases.

\section*{Acknowledgment}
The authors would like to thank Prof. Uri Erez at Tel Aviv University and Dr. Or Ordentlich at Massachusetts Institute of Technology for pointing out an error in the proof of the existence of Construction $\pi_A$ lattices that are good for MSE quantization in an earlier version of this paper.

\appendices

\section{Proof of Theorem~\ref{thm:goodness}}\label{apx:lattice_goodness}
We begin by noting that any lattice $\Lambda$ generated by Construction A can be written as (up to scaling) $\Lambda = \Lambda^* + p\mbb{Z}^N$, where $\Lambda^*$ is a coded level resulting from mapping a $(N,k)$ linear code to $\mbb{F}_p^N$ via a ring isomorphism and $p\mbb{Z}^N\defeq \Lambda'$ can be viewed as an uncoded level. As shown in \cite{forney2000}, one can first reduce the received signal by performing $\hspace{-3pt}\mod\Lambda'$. This will make the equivalent channel a $\Lambda/\Lambda'$ channel. When the underlying linear code is capacity-achieving for the $\Lambda/\Lambda'$ channel, the probability of error for the first level can be made arbitrarily small. Moreover, by choosing $p$ arbitrarily large, the probability that one would decode to a wrong lattice point inside the same coset can be made arbitrarily small. i.e., the probability of error for the second level can be made arbitrarily small. Forney \textit{et al.} in \cite{forney2000} showed the existence of a sequence of lattices that is good for channel coding under the above two conditions.

In the following, we closely follow the steps in \cite{forney2000} to show the existence of lattices that are good for channel coding generated by our construction. Let $p_1, p_2,\ldots, p_L$ be a collection of distinct odd primes. Similar to lattices from Construction A, a Construction $\pi_A$ lattice can be written as $\Lambda = \Lambda^* + \Pi_{l=1}^L p_l \mbb{Z}^N$ where $\Lambda^*$ is obtained from the steps 1) and 2) in Section~\ref{sec:prod_const} and $\Pi_{l=1}^L p_l \mbb{Z}^N \defeq \Lambda'$ is an uncoded level. Similar to \cite{forney2000}, the probability of error in the uncoded level can be made arbitrarily small when we choose $\Pi_{l=1}^L p_l$ sufficiently large. Therefore, one then has to show that the linear code $C^1\times\ldots\times C^L$ over $\mbb{F}_{p_1}\times\ldots\times \mbb{F}_{p_L}$ together with the mapping $\mc{M}$ is capacity-achieving for the $\Lambda/\Lambda'$ channel under multistage decoding.

Now, by the chain rule of mutual information \cite{cover91}, one has that
\begin{align}
    I(\msf{Y};\msf{X}) &= I(\msf{Y};\mc{M}(\msf{C}^1,\ldots,\msf{C}^L)) \nonumber \\
    &= I(\msf{Y};\msf{C}^1,\ldots,\msf{C}^L) = \sum_{l=1}^L I(\msf{Y};\msf{C}^l|\msf{C}^1,\ldots,\msf{C}^{l-1})).
\end{align}
Hence, the only task remained is showing that linear codes over $\mbb{F}_{p_l}$ can achieve the conditional mutual information $I(\msf{Y};\msf{C}^l|\msf{C}^1,\ldots,\msf{C}^{l-1})$. {\black Note that in \cite[Section III.A]{delsarte82}, it is shown that the average error probability $\bar{P}_e^{(N)}$ over the ensemble of random linear codes (form a balanced set) exponentially decays with $N$ for all rates smaller than the capacity if the channel is regular. If we randomly choose one code from this ensemble, by Markov inequality, we have
\begin{equation}
    \Pp(P_e^{(N)}\geq s\bar{P}_e^{(N)}) <\frac{1}{s}\defeq \epsilon,
\end{equation}
where $P_e^{N}$ is the probability of error and $s,\epsilon>0$. This guarantees that by randomly picking a code from this ensemble, with probability $1-\epsilon$, the error probability is not deviated too much from its average which is exponentially decayed in $N$.}

We now follow the proof in \cite{forney2000} and show that the equivalent channel at each level is regular in the sense of Delsarte and Piret \cite{delsarte82}. As restated in \cite{forney2000}, a channel with transition probabilities $\{f(y|b), b\in B, y\in Y\}$ is regular if the input alphabet can be identified with an Abelian group $B$ that acts on the output alphabet $Y$ by permutation. In other words, if a set of permutations $\{\tau_b, b\in B\}$ can be defined such that $\tau_b(\tau_{b'}(y))=\tau_{b\oplus b'}(y)$ for all $b, b'\in B$ and $y\in Y$ such that $f(y|b)$ depends only on $\tau_b(y)$. Note that since we are considering the $\Lambda/\Lambda'$ channel, the additive noise is actually the $\Lambda'$-aliased Gaussian noise given by
\begin{equation}
    f_{\Lambda'}(z) = \sum_{\boldsymbol\lambda\in\Lambda} g_{\eta^2}(z+\boldsymbol\lambda),~~z\in\mbb{R}^N,
\end{equation}
where $g_{\eta^2}(.)$ is the Gaussian density function with zero mean and variance $\eta^2$.

Now, suppose we are at the $l$th level's decoding. i.e., all the codewords in the previous levels have been successfully decoded. The receiver first subtracts out the contribution from the previous levels by $y-\mc{M}(c^1,\ldots,c^{l-1},0,\ldots,0)\hspace{-3pt}\mod\Lambda'$. We show that the equivalent channel seen at the $l$th level's decoding is regular. For $b\in\mbb{F}_{p_l}$ define
\begin{equation}
    \mathbf{b} \defeq \begin{bmatrix}
                        \mc{M}(0,\ldots,0,b,v_1^{l+1},\ldots,v_1^L) \\
                        \mc{M}(0,\ldots,0,b,v_2^{l+1},\ldots,v_2^L) \\
                        \vdots \\
                        \mc{M}(0,\ldots,0,b,v_S^{l+1},\ldots,v_S^L) \\
                      \end{bmatrix},
\end{equation}
where $(v_s^{l+1},\ldots,v_s^L)\in\mbb{F}_{p_{l+1}}\times\ldots\times\mbb{F}_{p_L}$ for $s\in\{1,\ldots,S\}$ and none of these vectors are exactly the same. Therefore, there are total $S=\Pi_{l'>l} p_{l'}$ possibilities. Also, note that the ordering of elements in $\mathbf{b}$ does not matter and can be arbitrarily placed. Thus, given the previously decoded codewords, $\mathbf{b}$ is fully determined by $b$. For $y\in\mbb{R}^N$, let us now define the following,
\begin{align}
    \tau_b(y)&\defeq y-\mathbf{b} \hspace{-3pt}\mod \Lambda' \nonumber \\
    &\defeq\begin{bmatrix}
                        y- \mc{M}(0,\ldots,0,b,v_1^{l+1},\ldots,v_1^L) \hspace{-3pt}\mod\Lambda'\\
                        y- \mc{M}(0,\ldots,0,b,v_2^{l+1},\ldots,v_2^L) \hspace{-3pt}\mod\Lambda' \\
                        \vdots \\
                        y- \mc{M}(0,\ldots,0,b,v_S^{l+1},\ldots,v_S^L) \hspace{-3pt}\mod\Lambda'\\
                      \end{bmatrix}.
\end{align}
One can verify that
\begin{align}
    &\tau_b(\tau_{b'}(y)) = \tau_{b'}(y) - \mathbf{b} \hspace{-3pt}\mod \Lambda' \nonumber \\
    &\overset{(a)}{=} \begin{bmatrix}
                        y- \mc{M}(0,\ldots,0,b'\oplus b,2v_1^{l+1},\ldots,2v_1^L) \hspace{-3pt}\mod\Lambda'\\
                        y- \mc{M}(0,\ldots,0,b'\oplus b,2v_2^{l+1},\ldots,2v_2^L) \hspace{-3pt}\mod\Lambda' \\
                        \vdots \\
                        y- \mc{M}(0,\ldots,0,b'\oplus b,2v_S^{l+1},\ldots,2v_S^L) \hspace{-3pt}\mod\Lambda'\\
                      \end{bmatrix} \nonumber \\
                      &= \begin{bmatrix}
                        y- \mc{M}(0,\ldots,0,b'\oplus b,\tilde{v}_1^{l+1},\ldots,\tilde{v}_1^L) \hspace{-3pt}\mod\Lambda'\\
                        y- \mc{M}(0,\ldots,0,b'\oplus b,\tilde{v}_2^{l+1},\ldots,\tilde{v}_2^L) \hspace{-3pt}\mod\Lambda' \\
                        \vdots \\
                        y- \mc{M}(0,\ldots,0,b'\oplus b,\tilde{v}_S^{l+1},\ldots,\tilde{v}_S^L)\hspace{-3pt}\mod\Lambda'\\
                      \end{bmatrix}, \label{eqn:pi_regular}
\end{align}
where $(\tilde{v}_s^{l+1},\ldots,\tilde{v}_s^L)\in\mbb{F}_{p_{l+1}}\times\ldots\times\mbb{F}_{p_L}$ for $s\in\{1,\ldots,S\}$ and (a) follows from the fact that $\mc{M}$ is an isomorphism. Now, since the mapping from $\mbb{Z}_p$ to $2\odot\mbb{Z}_p$ is bijective for all odd primes $p$, it is clear that none of $(\tilde{v}_s^{l+1},\ldots,\tilde{v}_s^L)$ for $s\in\{1,\ldots,S\}$ are the same so one can rearrange \eqref{eqn:pi_regular} to get $\tau_b(\tau_{b'}(y))=\tau_{b\oplus b'}(y)$.

Let $b\in\mbb{F}_{p_l}$ be transmitted, the transition probability is given by
\begin{align}
    f(y|c^1,\ldots,c^{l-1},b) &\propto  \nonumber \\
    \sum_{(v^{l+1},\ldots,v^L)\in\mbb{F}_{p_{l+1}}\times\ldots\times\mbb{F}_{p_L}} &f_{\Lambda'}(y|c^1,\ldots,c^{l-1},b,v^{l+1},\ldots,v^L),
\end{align}
which only depends on $\tau_b(y)$. Hence the equivalent channel experienced by the $l$th level is regular and linear codes suffice to achieve the mutual information. Repeating this argument to each level shows that multilevel coding and multistage decoding suffice to achieve the capacity.

\section{Lemma~\ref{lma:Vol_fine} and its Proof}\label{apx:Vol_fine}
\begin{lemma}\label{lma:Vol_fine}
    Let $\msf{Z}^*_{eq}$ be the i.i.d. Gaussian random vector having distribution $\mc{N}(0,\sigma^2_{eq})$ where $\sigma^2_{eq}$. There exists a sequence of fine lattices $\Lambda_f$ whose error probability can be made arbitrarily small under multistage decoding whenever
    \begin{equation}
        \text{Vol}(\Lambda)^{\frac{2}{N}}>2\pi e\sigma^2_{eq}2^{-\frac{2}{N}D(\msf{Z}_{eq}\|\msf{Z}^*_{eq})}.
    \end{equation}
\end{lemma}
\begin{IEEEproof}
Let $\Lambda$ be a lattice generated by Construction $\pi_A$ with primes $p_1$, $\ldots$, $p_L$ and let $\Lambda'$ be a sublattice of $\Lambda$. Define $C_{\text{U}}(\Lambda/\Lambda',\msf{Z_{eq}})$ and $C_{\text{U}}(\Lambda',\msf{Z_{eq}})$ the uniform input capacity of the $\Lambda/\Lambda'$ and $\hspace{-3pt}\mod$-$\Lambda'$ channels \cite{forney2000}, respectively, with noise distribution $\msf{Z}_{eq}$. We denote by $P_e(\Lambda',\msf{Z}_{eq})$ the error probability when using $\Lambda'$ over the channel with additive $\msf{Z}_{eq}$ noise. For a lattice $\Lambda$ and noise variance $\sigma^2_{eq}$, let us also define
\begin{equation}
    \alpha^2(\Lambda,\sigma^2_{eq}) \defeq \frac{\text{Vol}(\mc{V}_{\Lambda})^{\frac{2}{N}}}{2\pi e\sigma^2_{eq}}.
\end{equation}

Similar to \cite[Section V]{forney2000}, we begin with a lattice partition $\Lambda/\Lambda'$ such that
\begin{enumerate}
    \item $C_{\text{U}}(\Lambda/\Lambda',\msf{Z_{eq}})\approx C_{\text{U}}(\Lambda',\msf{Z_{eq}})$,
    \item $\text{Vol}(\mc{V}_{\Lambda'})$ is large enough that $P_e(\Lambda',\msf{Z}_{eq})\approx 0$,
\end{enumerate}
where the second condition is possible because $\msf{Z}_{eq}$ is semi norm-ergodic \cite{ordentlich_erez_simple} and requires $q\rightarrow\infty$.

Recall that $\msf{Z}_{eq}^*$ is a zero-mean Gaussian random vectors having a variance $\sigma_{eq}^2$. Consider the $\hspace{-3pt}\mod \Lambda'$ channel
\begin{equation}
    \mathbf{y}'=[\mathbf{x} + \mathbf{z}_{eq}]\hspace{-3pt}\mod\Lambda'.
\end{equation}
We have the uniform input capacity given by
\begin{align}\label{eqn:capa_achi_1}
    C_{\text{U}}(\Lambda',\msf{Z}_{eq}) &= I(\msf{Y}';\msf{X}) \nonumber \\
    &\overset{(a)}{=} \log\left(\text{Vol}(\mc{V}_{\Lambda'})\right)-h(\msf{Z}_{eq}\hspace{-3pt}\mod\Lambda') \nonumber \\
    &\geq \log\left(\text{Vol}(\mc{V}_{\Lambda'})\right)-h(\msf{Z}_{eq}) \nonumber \\
    &\overset{(b)}{=} \log\left(\text{Vol}(\mc{V}_{\Lambda'})\right)-h(\msf{Z}_{eq}^*) + D(\msf{Z}_{eq}||\msf{Z}_{eq}^*) \nonumber \\
    &= C_{\text{U}}(\Lambda',\msf{Z}_{eq}^*) + D(\msf{Z}_{eq}||\msf{Z}_{eq}^*) \nonumber \\
    &\overset{(c)}{\approx} \frac{N}{2}\log\alpha^2(\Lambda',\sigma_{eq}^2)+D(\msf{Z}_{eq}||\msf{Z}_{eq}^*),
\end{align}
where (a) follows from the crypto lemma, (b) is due to the fact that $D(\msf{Z}_{eq}||\msf{Z}_{eq}^*) = h(\msf{Z}_{eq}^*)-h(\msf{Z}_{eq})$ \cite{cover91},
and (c) is from \cite[Theorem 3 and Theorem 10]{forney2000} that $C_{\text{U}}(\Lambda',\msf{Z}_{eq}^*)=C(\Lambda',\msf{Z}_{eq}^*)\approx \frac{N}{2}\log\alpha^2(\Lambda',\sigma_{eq}^2)$ the true capacity of the $\hspace{-3pt}\mod\Lambda'$ channel with noise $\msf{Z}_{eq}^*$.

By the first assumption above, one has
\begin{align}\label{eqn:capa_achi_2}
    \frac{2}{N}C_{\text{U}}(\Lambda',\msf{Z_{eq}})&\approx \frac{2}{N}C_{\text{U}}(\Lambda/\Lambda',\msf{Z_{eq}}) \nonumber \\
    &\overset{(c)}{\approx} \frac{2}{N}\log\left(\frac{\text{Vol}(\mc{V}_{\Lambda'})}{\text{Vol}(\mc{V}_{\Lambda})}\right) \nonumber \\
    &= \log\alpha^2(\Lambda',\sigma_{eq}^2) - \log\alpha^2(\Lambda,\sigma_{eq}^2),
\end{align}
where (c) is because the underlying linear codes are capacity-achieving. Combining \eqref{eqn:capa_achi_1} and \eqref{eqn:capa_achi_2} results in
\begin{equation}
    \text{Vol}(\mc{V}_{\Lambda})^{\frac{2}{N}} \approx 2\pi e\sigma^2_{eq}2^{-\frac{2}{N}D(\msf{Z}_{eq}||\msf{Z}^*_{eq})}.
\end{equation}

The error probability can be union bounded as
\begin{equation}
    \Pp(\text{errors in the coded levels}) + \Pp(\text{errors in the uncoded level}),
\end{equation}
which, similar to \cite{forney2000}, can be made arbitrarily small since the code is capacity-achieving and $\text{Vol}(\mc{V}_{\Lambda'})$ is large enough to avoid errors in the uncoded level. Moreover, similar to Appendix~\ref{apx:lattice_goodness}, one can use the chain rule to show that $C_\text{U}(\Lambda',\msf{Z}_{eq})$ can be achieved with multilevel coding and multistage decoding.
\end{IEEEproof}

\end{document}